\documentclass[showpacs, amsmath, amssymb, amsthm, aps, jmp, thmsb, twocolumn, superscriptaddress, pra]{revtex4-2}

\usepackage[dvipdfmx]{graphicx}
\usepackage{color}
\usepackage{float}
\usepackage{amsmath}
\usepackage{amssymb}
\usepackage{amsfonts}
\usepackage{amsbsy}
\usepackage{latexsym}
\usepackage{mathrsfs}
\usepackage{bbm}
\usepackage{bm}
\usepackage{booktabs}
\usepackage{multirow}

\newtheorem{definition}{Definition}

\newtheorem{theorem}{Theorem}

\newenvironment{proof}[1][Proof]{\textbf{#1.} }{\ \rule{0.5em}{0.5em}}
\newcommand{\be}{\begin{eqnarray}}
\newcommand{\ee}{\end{eqnarray}}
\def\({\left(}
\def\){\right)}
\def\[{\left[}
\def\]{\right]}
\def\C{\mathbb{C}}
\def\R{\mathbb{R}}
\def\Z{\mathbb{Z}}
\def\N{\mathbb{N}}

\newcommand{\bra}[1]{\langle #1 |}
\newcommand{\ket}[1]{| #1 \rangle}

\newcommand{\sla}[1]{\rlap{\kern .15em /}#1}

\begin{document}
%\title{Can efficiently calculable randomness measures distinguish quantum randomness \\ from pseudo-randomness?}
%\title{Relative frequency and algorithmic complexity in quantum bit strings:\\ an empirical study}
%\title{Demystifying a possible relation between non-locality and quantum randomness: Empirical data reconsidered}
%\title{Randomness and Violation of Bell Inequality: Reconsidering Empirical Data}
\title{Indistinguishability between quantum randomness and pseudo-randomness\\
under efficiently calculable randomness measures}

\author{Toyohiro Tsurumaru}
\email{Tsurumaru.Toyohiro@da.MitsubishiElectric.co.jp}
\affiliation{Mitsubishi Electric Corporation, Information Technology R\&D Center, Kanagawa, 247-8501, Japan.
}

\author{Tsubasa Ichikawa}
\email{ichikawa.tsubasa.qiqb@osaka-u.ac.jp}
\affiliation{Center for Quantum Information and Quantum Biology (QIQB), Osaka University, Osaka 560-0043, Japan.
}

\author{Yosuke Takubo}
\email{yosuke.takubo@kek.jp}
\affiliation{Institute of Particle and Nuclear Studies, High Energy Accelerator Research Organization (KEK), Ibaraki 305-0801, Japan.
}
\affiliation{The Graduate University for Advanced Studies (SOKENDAI), Hayama 240-0193, Japan.
}

\author{Toshihiko Sasaki}
\email{sasaki@qi.t.u-tokyo.ac.jp}
\affiliation{Photon Science Center, Graduate School of Engineering, The University of Tokyo, Tokyo 113-8656, Japan.
}

\author{Jaeha Lee}
\email{lee@iis.u-tokyo.ac.jp}
\affiliation{Institute of Industrial Science, The University of Tokyo, Chiba 277-8574, Japan
}

\author{Izumi Tsutsui}
\email{izumi.tsutsui@kek.jp}
\affiliation{Institute of Particle and Nuclear Studies, High Energy Accelerator Research Organization (KEK), Ibaraki 305-0801, Japan.
}
%\affiliation{Department of Physics, The University of Tokyo, Tokyo 113-0033, Japan.
%}

\begin{abstract}
We present a no-go theorem for the distinguishability between quantum random numbers ({\it i.e.}, random numbers generated quantum mechanically) and pseudo-random numbers ({\it i.e.}, random numbers generated algorithmically).
The theorem states that one cannot distinguish these two types of random numbers if the quantum random numbers are efficiently classically simulatable and the randomness measure used for the distinction is efficiently computable.
We derive this theorem by using the properties of cryptographic pseudo-random number generators, which are believed to exist in the field of cryptography.
Our theorem is found to be consistent with the analyses on the actual data of quantum random numbers generated by the IBM Quantum and also those obtained in the Innsbruck experiment for the Bell test, where 
the degrees of randomness of these two set of quantum random numbers turn out to be essentially indistinguishable from those of the corresponding pseudo-random numbers.   
Previous observations on the algorithmic randomness of quantum random numbers are also discussed and reinterpreted in terms of our theorems and data analyses.
\end{abstract}
%\pacs{%03.65.Vf, 03.67.Pp, 82.56.Jn.
%}

\maketitle

%%%%%%%%%%%%%%%%%%%%%%%%%%%%%%
\section{Introduction}
%%%%%%%%%%%%%%%%%%%%%%%%%%%%%%

In quantum mechanics, it is widely believed that transitions between two states occur randomly, and that this randomness is genuine in the sense that it does not admit causal descriptions by means of determinants called hidden variables.
This belief has been supported by numerous experiments (Bell tests) \cite{PhysRevLett.28.938,Aspect:1982,Tittel:1998,PhysRevLett.81.5039,Rowe:2001,Sakai:2006,Ansmann:2009,Giustina:2013,Giustina:2015,Hensen:2015,Shalm:2015,Dehollain:2016,Rosenfeld:2017,Rauch:2018} to examine the Bell inequality \cite{Bell:1964} which ended up in excluding any local hidden variable theories.  
Besides, no experiments so far signal even non-local hidden variable theories including Bohm\rq s one \cite{PhysRev.85.166,PhysRev.85.180}, 
suggesting the implausibility of causal descriptions underlying the quantum randomness.  

Meanwhile, the genuineness of quantum randomness has become one of the bases for applications of quantum physics, such as in quantum cryptography \cite{Pirandola:20} and random number generation \cite{RevModPhys.89.015004}.
It is also considered crucial in quantum mechanics to uphold the no-cloning theorem and ensure the consistency between quantum non-locality and special relativity \cite{10.5555/2666124}.

In view of the firmly established foundational trait and the ongoing successful applications, it is somewhat surprising that we still have little  understanding on how different the quantum randomness is with respect to non-quantum ones.   
Recently, we have seen several attempts \cite{PhysRevA.82.022102,Solis_2015,e20110886,PhysRevA.98.042131,Abbott_2019} designed to characterize the quantum randomness quantitatively
using some randomness measures based on 
the theory of algorithmic complexity \cite{Cover2006} and/or the statistical homogeneity of distribution.  
The algorithmic complexity quantifies the upper bounds of the memory size to generate the bit string with a universal Turing machine, whereas the statistical homogeneity evaluates relative frequency of the binary values of the bit string.  

Among the randomness measures mentioned in these attempts, the Lempel-Ziv (LZ) complexity and the Borel normality appear particularly handy and commonly used.  Here,  
the LZ complexity \cite{1055501}, which has been utilized to evaluate the complexity of non-linear dynamical systems in physics \cite{PhysRevA.36.842,10.1063/1.4808251}, estimates the difficulty of compressing the given bit string with  the LZ76 algorithm  \cite{1055501}.  On the other hand, the Borel normality \cite{calude2002}, which has been employed to compare the randomness of quantum systems with those of the other systems \cite{PhysRevA.82.022102,Solis_2015,e20110886,PhysRevA.98.042131,Abbott_2019}, 
quantifies the difference between the distribution of substrings in the given bit string and the uniform distribution thereof.

With these two measures, previous data analyses led us to two intriguing observations:
\begin{itemize}
\item[i)] The bit strings generated from a quantum random number generator (QRNG) have statistically larger values of the Borel normality, which implies that the generated bit strings are less random than those generated from non-quantum pseudo-random number generators (PRNGs) such as the Mersenne twister \cite{Abbott_2019}, despite the fact that by definition the latter is actually less random than the former (see Section \ref{sec:PRNG} for detail). This has been argued to be caused by the effects of unknown biases coming from experimental imperfections for the implementation of the QRNG.
\item[ii)] The bit strings generated from the Bell tests (another type of QRNGs) in \cite{PhysRevLett.81.5039,weihs_gregor_1998_7185335} appear to have smaller values of the LZ complexity when the corresponding Bell tests violate the Bell inequality more strongly \cite{PhysRevA.98.042131}, which again implies that the generated bit strings are less random when they are generated non-locally.  In other words, the degree of algorithmic randomness
appears to be anti-correlates with the degree of non-locality.
\end{itemize}

The above observations are clearly in conflict with our belief on the genuineness of quantum randomness and allure us to doubt the validity of the randomness measure employed, or even to question the implicit assumption on the very existence of measures capable of  distinguishing QRNGs from PRNGs, possibly under the influence of experimental imperfections.

In fact, it is not difficult to infer that the above assumption may be untenable.
For example, suppose that one is given equal-weight superposition states of quantum two-level systems (qubits).
By repeating ideal projection measurements, one will obtain a series of genuine random outcomes of binary numbers and, hence, a qubit may be regarded as the  \lq true random number generator\rq~(TRNG)  (our more technical definition of the TRNG is given in Sec.~\ref{theory}).
On the other hand, PRNGs for commercial use pass the randomness tests such as NIST suites, and are believed to imitate the TRNG with reasonable precisions, %despite the fact that by definition PRNGs are less random than the TRNG, 
as far as the bit strings are finite and randomness tests are implementable in finite durations.
This seems to contradict the fact that by definition PRNGs are less random than the TRNG.
This indicates that, in practice, PRNGs will be able to mimic QRNGs at least for an ideal case.

On the basis of the above consideration, in this paper we present a no-go theorem on the above assumption.
We derive this theorem by using the properties of cryptographic PRNGs (CPRNGs; see, e.g., Refs.\cite{books/crc/KatzLindell2020,goldreich_2001}), which are believed to exist in the field of cryptography.
To be more explicit, 
we shall show that, if there exists a CPRNG, whose outcomes are indistinguishable from those of the TRNG (under any polynomial time algorithm within a negligible error margin), then any efficiently computable randomness measures cannot distinguish the outcomes of systems which are \lq efficiently classically simulatable\rq\ \cite{doi:10.1098/rspa.2002.1097} (which do not necessarily behave as the TRNG) from those of the PRNGs with a negligible error margin.
Since some quantum systems are known to be efficiently classically simulatable, our no-go theorem is applicable to those systems.  We also provide examples where the counter-intuitive observations i), ii) can be understood more easily as a consequence of the PRNGs capable of mimicking the QRNGs.  This illustrates that the observations i), ii) are actually consistent with our theorem.

In concrete terms, for observation i) we find that, if the length of the bit string and the relative frequency of bit values are comparable between the dataset generated by IBM Quantum and those generated by PRNGs, then the distributions of the LZ complexity and Borel normality are almost identical between IBMQ and PRNGs.
The difference between them reported in \cite{Abbott_2019} can be attributed to the fact that the relative frequencies were not aligned properly.

For observation ii), we first point out flaws in the data analyses in \cite{PhysRevA.98.042131} (see the last two paragraphs of Section \ref{sec:LZ_complexity}).
 Motivated by this, we reanalyze the dataset obtained in the Innsbruck experiment for the Bell test \cite{weihs_gregor_1998_7185335} and compare the results with those obtained from PRNGs.
As with the quantum coin tosses, this analysis also shows that the distributions of the LZ complexity and Borel normality are comparable in the actual dataset and PRNGs.
In addition, we reject the hypothesis of the existence of an anti-correlation between the degree of randomness and that of the violation of the Bell inequality.
We shall also exhibit a correlation between the LZ complexity and the binary entropy of the relative frequency, which indicates that systematic errors latent in the data acquisition process influence the degree of randomness.

This paper is organized as follows:
In Sec.~\ref{theory}, we introduce notions relevant to our theoretical analyses and present our main results including the no-go theorem.
In Sec.~\ref{supports}, we make the empirical analyses mentioned above to support our results and argue their implications and interpretations with respect to preceding works.
Section~\ref{conc} is devoted to our conclusion.

\section{Our notion of randomness and the main results}
\label{theory}

%\textcolor{blue}{
As announced in Introduction, we show the indistinguishability of the outcomes of the QRNGs and PRNGs with the use of any efficiently calculable randomness measures.
Proof of this assertion requires several theoretical notions originating from cryptography, but possibly unusual in quantum information science.
We hereafter make a review of these important notions (with slight generalizations if necessary), rigorously state our assertion, and give proof thereof.
%}

We begin by summarizing the terminology. 
A quantum random number generator (QRNG) is literally a quantum system that outputs a random value.
A pseudorandom number generator (PRNG), on the other hand, is a deterministic computation algorithm which transforms a short random bit sequence into a seemingly random long bit sequence (see Section \ref{sec:PRNG} for detail).

References \cite{1055501,calude2002} discuss indices $I$ which can be computed from a random bit sequence $X$ and measure the randomness of $X$.
Hereafter such index $I(X)$ will also be called a {\it randomness measure}.
The references above also suggest the possibility that certain types of $I(X)$, such as the LZ complexity (see Section \ref{sec:LZ_complexity}) and the Borel normality (see Section \ref{sec:Borel_normality}), can be used to distinguish PRNGs and QRNGs.

Our main message in this paper is that such an index $I$ is not in fact feasible, in light of the knowledge of modern cryptography.
The basic flow of discussion is as follows.

While references \cite{1055501, calude2002} placed no particular constraints on $I$, in the real world, one needs at least the following two conditions, for $I$ to be able to distinguish PRNGs and QRNGs,
\begin{enumerate}
\item There is a function $D_I:I\to\{0,1\}$ that maps index $I$ to one bit.
The distribution of $D_I(I(X))$ differs depending on whether the random number $X$ is output from QRNG or PRNG.
\item Indices $I, D_I$ can be calculated in a realistic time.
\end{enumerate}
(See Section \ref{sec:minimal_requirements_I}, Definitions \ref{def:randomness_measure} and \ref{def:necessary_condition_on_I} for details.)
However, if any randomness measure $I$ can satisfy these conditions, it immediately leads to a contradiction with the common and widely used assumption of cryptography, namely, the existence of cryptographic PRNGs.
In other words, as long as we accept the common assumption of cryptography, we can show that there is no good randomness measure $I$ that can distinguish between QRNG and PRNG.

The meanings of conditions 1 and 2 above are as follows.

First, as a major premise, if an index $I$ can actually distinguish QRNGs and PRNGs, its value $I(X)$ must have different probability distributions depending on whether $X$ is from a QRNG or a PRNG.
Or equivalently, if $I(X)$ has exactly the same distribution for the two sources, then $I$ will never be able to distinguish them.

Under this premise, condition 1 above demands that there exists a rule $D_I$ that translates such difference in distributions of $I(X)$ to one bit, and that the one bit can be used to classify the sources into two groups.
Ideally, $D_I$ should be strong enough to discriminate QRNGs and PRNGs without error, but that is not necessarily required here.
It here suffices that the distribution of $D_I(I(X))$ exhibits a slight difference depending on whether $X$ is from a QRNG or a PRNG (see Section \ref{sec:distinguisher} for details).

Condition 2 demands that the computation necessary for such classification can be carried out within a realistic time.
For example, if the calculation of a given type of $I$ or $D_I$ takes an infinite, or a finite but unrealistically long time (e.g., longer than the age of the universe even using all the computers on earth), then such $I$ or $D_I$ are considered incalculable in practice, and thus excluded them from consideration.

As to the rigorous notion of `can be calculated in a realistic time,' we follow the tradition of computer science and define it as `computable by an algorithm (Turing machine) in time which is a polynomial of the input length' \cite{sipser13}.
Throughout the paper, we will often call a polynomial-time algorithm an {\it efficient} algorithm.
Also, unless otherwise stated, whenever we say simply an `algorithm,' it refers to an efficient algorithm.

\subsection{Facts on random number generators (RNGs)}

\subsubsection{True random number generator (TRNG)}
\label{sec:TRNG}

Throughout the paper, $P(E)$ denotes the probability of an event $E$.

For each $n\in \N$, we denote by $U^n$ the random variable of the uniformly random $n$ bits; i.e. $U^n\in\{0,1\}^n$, and $P(U_n=x)=2^{-n}$ for all $x\in\{0,1\}^n$.
This random variable $U^n$ achieves the maximum Shannon entropy possible for each $n$,
\begin{equation}
H(U^n)=n
\label{eq:TRNG_entropy_maximum}
\end{equation}
(where  $H(X)$ denotes the Shannon entropy of a random variable $X$; $H(X)=-\sum_{x}p(x)\log_2 p(x)$, $p(x)=P(X=x)$. See, e.g., Ref. \cite{Cover2006}).
For this reason, we will call $U_n$ the true random number generators (TRNGs).

\subsubsection{The TRNGs are a special case of QRNGs}

The TRNGs are a special case of QRNGs.
That is, the TRNGs can be realized in a quantum system in principle.

The case of $n=1$, $U^1$, can be realized because there exists a quantum system $A$ that outputs bit values 0,1 exactly with probability 1/2, which realizes.
For example, let $A$ be a quantum system where one (i) generates one of the $Z$ basis states $\{\ket{0},\ket{1}\}$, and (ii) performs a projective measurement on it using the $X$ basis $\{\ket{\tilde{0}},\ket{\tilde{1}}\}$, with $\ket{\tilde{0}}:=(\ket{0}+(-1)^b\ket{1})/\sqrt2$.
Also, if we let $A^n$ denote an $n$ repetition of a system $A$, then $A^n$ realize $U^n$ for an arbitrary $n\in\N$ (Fig. \ref{fig:randomness1}(a)).

At this point, one might argue that in a real physical system there is always a limitation of experimental accuracy, and it might not be possible to achieve probability 1/2 strictly. 
Even if it were possible, one may not be able to repeat an arbitrary number of trials using exactly the same system, due to limitations in costs and materials.

Still, it should be noted that, in theory, one cannot deny the existence of such $A$ as well as its iterations $A^n$.
That is, no rationale can exclude the TRNG from QRNGs.
Hence, for example, if one claims that an index $I$ can distinguish QRNGs from PRNGs, one should also assume that it can distinguish the TRNG, a type of QRNG, from PRNGs.

\begin{figure}[t]
  \includegraphics[width=\linewidth]{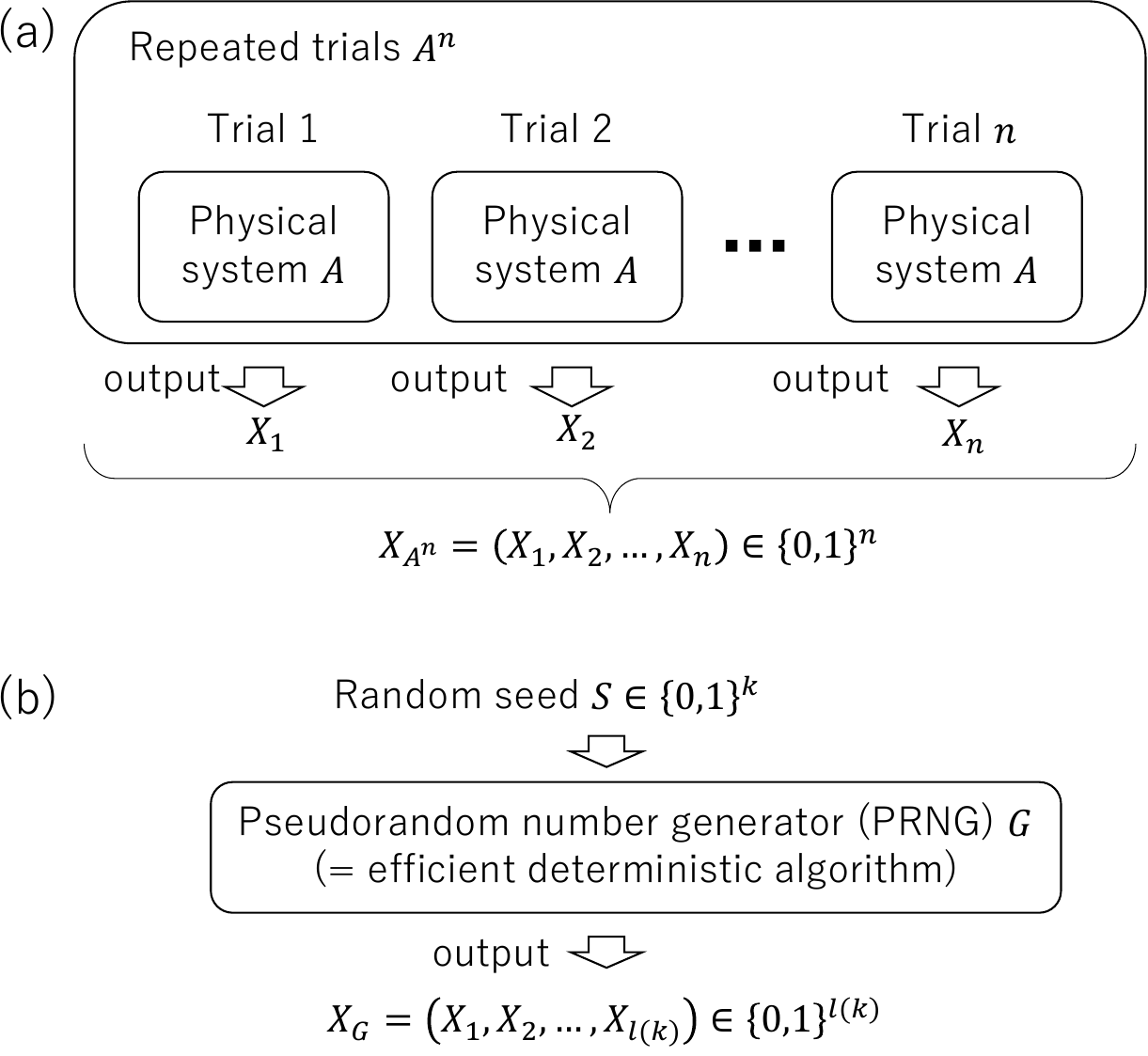}
  \caption{(a) Repeated trials $A^n$ of experiments using the same physical system $A$.
(b) Pseudorandom number generator (PRNG) $G$.
  }
  \label{fig:randomness1}
\end{figure}

\subsubsection{Pseudorandom number generator (PRNG)}
\label{sec:PRNG}

Informally, a pseudorandom number generator (PRNG) is
\begin{description}
\item[P1] A method for expanding a short random bit string $s$ (called a {\it randm seed}) into a long bit string $x$ 
\end{description}
(cf. Fig. \ref{fig:randomness1}(b) and Definition \ref{def:PRNG}), where
\begin{description}
\item[P2] The generated string $x$ is in some sense similar to a true random bits string (i.e. an output of the TRNG).
\end{description}

Condition P1 means that, once one feeds a short random seed $s$ to a PRNG, one can obtain a virtually inexhaustible number of `random' bits $x$.
PRNGs are useful in this sense in practive, but the drawback is that the resulting $x$ is much less random than TRNG.

Condition P2 means that for certain usages the small randomness of $x$ may not be apparent, and $x$ can be a substitute for the TRNG.
For example, in numerical calculations, it is observed that using PRNG instead of TRNG does not significantly change the results.

Condition P1 is stated formally as follows.
\begin{definition}[PRNG, cf. Fig. \ref{fig:randomness1}(b)]
\label{def:PRNG}
Pseudorandom number generator (PRNG) is a (efficient) deterministic algorithm such that for any $k\in\N$ and any input $s\in\{0,1\}^k$, the result $G(s)$ is a string of length $l(k)$, where $l(k)$ is a polynomial satisfying $l(k)>k$ for $\forall k\in\N$.
The polynomial $l(k)$ is called an expansion factor.
\end{definition}

Here a deterministic algorithm $G$ is the one in which, once the input value $s$ is determined, the output value $x$ is also uniquely determined.

Since the random seed $S$ is a $k$-bit long random variable, its entropy is at most $k$ bits.
Also, since the algorithm $G$ deterministic, the entropy of its output $X=G(S)$ can at most be that of $S$.
Thus, we have
\begin{equation}
H(G(S))\le k<l(k),
\end{equation}
meaning that, unlike in the case of the TRNG, the entropy of PRNG's output $X=G(S)$ can never attain the theoretical maximum $l(k)$ (cf. Eq. (\ref{eq:TRNG_entropy_maximum})).
PRNGs are definitely inferior to the TRNGs in this sense.

\subsubsection{Cryptographic PRNG (CPRNG)}
\label{sec:CPRNG}

Next we proceed to the formal definition of condition P2.
Unlike condition P1, the definition of P2 differs depending on the academic field.
Here we particularly adopt the definition in the field of cryptography, or {\it cryptology}.

In cryptology, the phrase `in some sense similar' in condition P2 is taken pessimistically and interpreted as `distinguishable by any efficient discrimination algorithm.'
This is because rationale cryptanalysts will always employ the best attacking method possible.
In other words, no cryptanalysts will deliberately choose an inferior method when better methods are available.

From this point of view, condition P2 can formally be stated as follows.
\begin{definition}[CPRNG. Ref. \cite{books/crc/KatzLindell2020}, Def. 3.14]
\label{def:Cryptographic_PRNG}
We say that a PRNG $G$ is a cryptographic PRNG (CPRNG) if it satisies, in addition to the properties of Definition \ref{def:PRNG}, the following:
For any (efficient) probabilistic algorithm $D$, there is a negligible function ${\rm negl}$ such that
\begin{align}
\left|P(D(G(U^k))=1)-P(D(U^{l(k)})=1)\right|\le{\rm negl}(k).
\end{align}
\end{definition}
(Several remarks are in order for this definition.
First, probabilistic algorithms are those which perform coin tosses internally (see, e.g., \cite{sipser13, goldreich_2001}).
Second, a function $f:\N\to\R$ is called negligible if it satisfies: For any positive polynomial (polynomial whose output is positive for any input $\in\N$) $p(k)$, there exists $N\in\N$ such that
\begin{equation}
\forall k\ge N,\ f(k)<\frac{1}{p(k)};
\end{equation}
see, e.g., Ref. \cite{goldreich_2001}, Def. 1.3.5.
Finally, the distinguisher algorithm $D$ may vary depending on the PRNG algorithm $G$; in other words, this definition covers the situation where the cryptanalysts who try to distinguish between the PRNG and the TRNG {\it know} the PRNG algorithm.)

In fact, however, no explicit construction of CPRNG is known, which satisfies the properties of Definition \ref{def:Cryptographic_PRNG}. No proof of its existence has been given either.
Hence the existence of a CPRNG is a mere assumption, not a proven fact.

In such a situation, what is mostly done in cryptology is either (i) to mathematically prove the existence of a CPRNG from other assumptions, or (ii) to mathematically prove the security of a cryptosystem having $G$ as a component, assuming that $G$ is a CPRNG; see, e.g., Refs.\cite{books/crc/KatzLindell2020,goldreich_2001}.
As a result, most of the results obtained so far in cryptology do not hold without the existence of a CPRNG or other stronger assumptions, such as the existence of a one-way function.

It may be easier to understand this situation by considering CPRNG in cryptology as an analog of the second law in thermodynamics.
In thermodynamics, there are two situations similar to (i) and (ii) above, namely,
(i') if one accepts the second law as a hypothesis, one can prove useful results including the existence of entropy, and (ii') if one accepts the principles of statistical dynamics as hypotheses, one can explain to some extent (though one cannot prove) why the second law holds.

Hence, although the existence of a CPRNG is nothing more than a hypothesis, if any result contradicting with it were obtained, it would immediately mean the collapse of cryptology as a whole, and the impact would be enormous.

In terms of practicality, this can be stated as follows. Many standardized cryptographic schemes that we currently use specify a PRNG to be employed as their component. For example,  a cryptographic scheme called ChaCha20-Poly1305 employs a PRNG called ChaCha20 \cite{rfc8439, books/crc/KatzLindell2020}. %, which is used in our analyses in Section \ref{supports}. 
When theoretically guaranteeing the security of these standardized cryptographic schemes, it is customary to simply assume (rather than prove) that the PRNG is a CPRNG (or a pseudorandom function, which can be utilized to construct a CPRNG; see Section 3.5.1, Ref. \cite{books/crc/KatzLindell2020}), and then to ``prove'' that the scheme using it is secure based on the assumption (see, e.g., Section 3.2, Ref. \cite{books/crc/KatzLindell2020}). Therefore, if the existence of a CPRNG were denied,  the basis for theoretically guaranteeing the security of standardized cryptographic schemes would be lost.

\subsection{Minimal requirements on randomness measures}
\label{sec:minimal_requirements_I}

References \cite{1055501,calude2002} discuss randomness measures for distinguishing QRNGs and PRNGs.
There, a randomness measure $I$ is an index $I(X)$ that can be calculated from the output bit string $X\in\{0,1\}^n$ of a RNG, and reflects the randomness of the RNG.
However, they did not necessarily discuss conditions for $I$ to be useful in practice.
Below  will discuss such conditions.

\subsubsection{Efficiently computable randomness measure}
Some indices $I$, such as LZ complexity (see Section \ref{sec:LZ_complexity}), take a discrete value, while others, such as Borel normality (see Section \ref{sec:Borel_normality}), take a continuous value.
Those indices $I$ of discrete value can of course be represented by a finite number of bits.
On the other hand, those of a continuous value cannot be calculated as it is in finite time, hence one needs to adapt them to a finite number of bits, e.g., by truncating the value to a certain number of significant digits.
As a result, any index $I$ needs to be of finite bit length.

Moreover, in order for the index $I$ to be calculable not only in finite time but also in a realistic time, we need an efficient algorithm that can compute it.
Hence we need the following condition on $I$.
\begin{definition}[Efficiently computable measure $I$]
\label{def:randomness_measure}
A randomness measure is an efficient probabilistic algorithm
$I:\{0,1\}^*\to\{0,1\}^*$.
\end{definition}
Here $\{0,1\}^*$ denotes the set of all bit strings of a finite length, $\{0,1\}^*=\bigcup_{n\in \Z_{\ge0}}\{0,1\}^n$ (see, e.g., Refs.\cite{books/crc/KatzLindell2020,goldreich_2001}).

Note that Definition \ref{def:randomness_measure} can accommodate the case where one handles multiple indices $I_1, I_2, \dots$ simultaneously:
Let $I$ be a concatenating $I_1, I_2, \dots$, for example.

\subsubsection{Efficient distinguisher for a randomness measure}
\label{sec:distinguisher}

Furthermore, for $I$ to be an effective randomness measure, it is necessary that `$I$ reflects the randomness of the RNG.'
In this paper we interpret this property as that `the value $I(X)$ has different probability distributions
%when $X$ is from RNGs of different levels of randomness 
depending on whether the source of $X$ is a TRNG or a PRNG.'
Indeed, if this is not the case, then $I(X)$ will have the same distribution irrespective of the randomness of $X$'s source.
It is apparent that such $I(X)$ conveys no information on the randomness of $X$.

In addition, in order for humans to be able to exploit such differences in probability distribution, a quantitative criterion is necessary.
That is, one needs an algorithm $D_I$ which determines from the value of $I(X)$ whether the randomness of $X$'s source is good or bad (0 or 1).
Also, in order for $D_I$ to be calculable in a realistic time, $D_I$ needs to be efficient.

In summary, for any effective randomness measure $I$, there should exist an algorithm $D_I$ of the following type.
\begin{definition}[Distinguisher $D_I$ for measure $I$]
\label{def:necessary_condition_on_I}
A distinguisher for a randomness measure $I$ is an efficient probabilistic algorithm
$D_I:\{0,1\}^* \to \{0,1\}$.
\end{definition}
Also, we say that $I$ is a good randomness measure if the distribution of $D_I(I(X))$ is significantly different depending on whether the source of $X$ is a QRNG or a PRNG.

Under these definitions, the ideal situation is, of course, where the algorithm $D_I\circ I$ can discriminate QRNGs and PRNGs perfectly without any error.
However, it should be noted that we do not necessarily require such ideal discrimination here.

For example, $I$ can be regarded a good randomness measure as long as the result $D_I(I(x))=1$ is reliable, even if $D_I(I(x))=0$ is not.
This corresponds to the following common situation in randomness tests (such as Ref. \cite{8966}):
If $x$ fails the test (i.e., if $D_I(I(x))=1$) one can safely conclude the source of $x$ is a bad one (in our case, a PRNG), but even if $x$ passes the test (i.e., if $D_I(I(x))=0$) it does not necessarily mean that the source was a good one (in our case, a QRNG).

\begin{figure}[t]
  \includegraphics[width=0.8\linewidth]{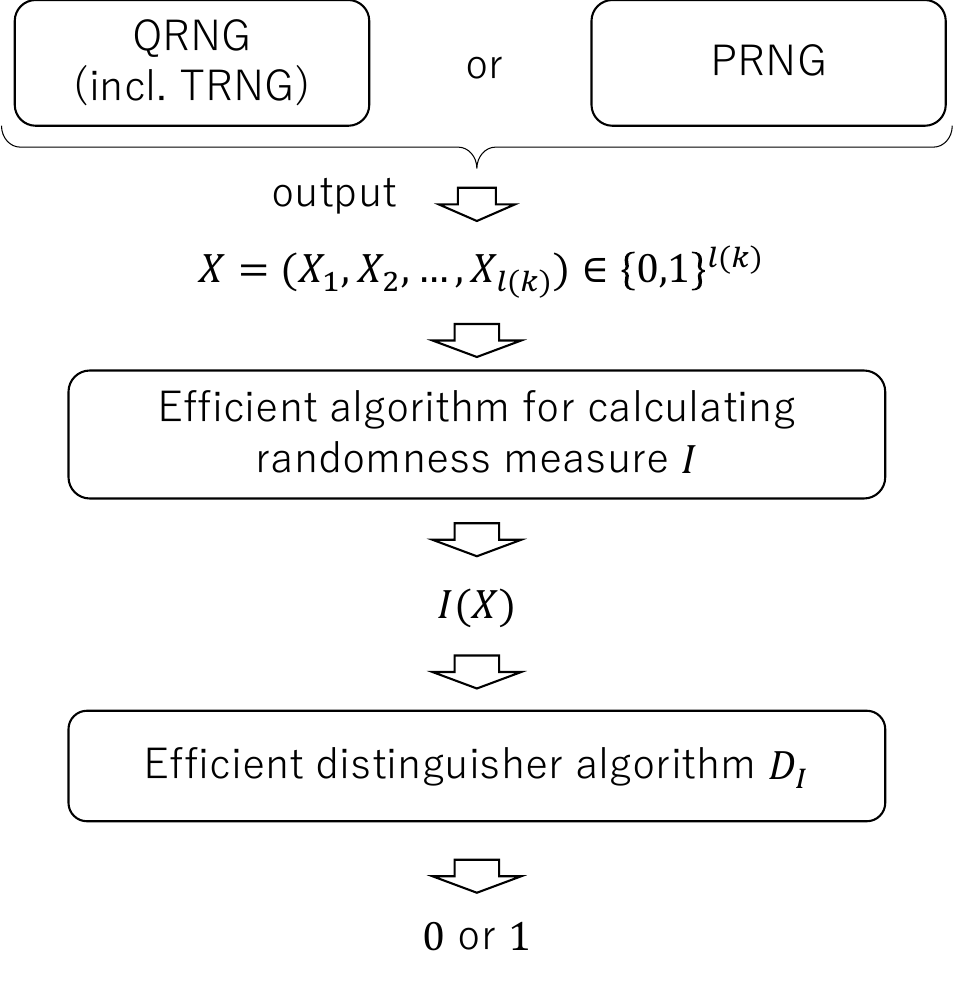}
  \caption{A setting for discriminating QRNGs (including the TRNG) from PRNGs using a randomness measure $I$ and a distinguisher $D_I$ corresponding to $I$.
  }
  \label{fig:randomness2}
\end{figure}

\subsection{Main result: A good randomness measure is not feasible}

In fact, the existence of a good randomness measure $I$ directly contradicts the common assumption of cryptology.
Hence a good randomness measure is not feasible.

More precisely, from what we have already discussed, one can conclude that
\begin{itemize}
\item[(i)] If a cryptographic PRNG (CPRNG), defined in Section \ref{sec:CPRNG}, exists, it is indistinguishable from the TRNG.
\item[(ii)] If there exists the algorithm $D_I\circ I$ that can distinguish QRNGs and PRNGs, it should also be able to distinguish the TRNG (a type of QRNG) and CPRNGs (a type of PRNGs).
\end{itemize}
However, these statements are clearly contradictory, thus one must give up either one.
That is, one needs to accept either of the followings statements,
\begin{itemize}
\item[(i')] A CPRNG does not exist.
\item[(ii')] No algorithm $D_I\circ I$ corresponding to any randomness measure $I$ can distinguish between the TRNG (a type of QRNG) from CPRNGs (a type of PRNGs).
\end{itemize}

Recall that, as mentioned in the second half of Section \ref{sec:CPRNG}, the existence of CPRNG is one of the most common assumptions in cryptology, on which most results of the field are based.
Thus, choosing (i') is equivalent to making a very challenging claim that denies much of cryptology.

Therefore, from a conservative standpoint, statement (ii') should be accepted.
The rigorous form of this statement is as follows.
\begin{theorem}
\label{thm:main}
Suppose that there exists a CPRNG $G$, then for any randomness measure $I$ and any distinguisher $D_I$ corresponding to it, we have
\begin{equation}
\left|P(D_I(I(G(U^k)))=1)-P(D_I(I(U^{l(k)}))=1)\right|\le{\rm negl}(k),
\label{eq:main_th_CPRNG}
\end{equation}
where $l(k)$ is the expanding factor of $G$
(Recall that $P(E)$ denotes the probability of an event $E$).
\end{theorem}

\begin{proof}
This theorem follows trivially from Definition \ref{def:Cryptographic_PRNG} of CPRNG.
Indeed, for any given $I$ and $D_I$ satisfying Definitions \ref{def:randomness_measure} and \ref{def:necessary_condition_on_I}, the algorithm $D_I\circ I$ becomes an example of the distinguisher algorithms $D$ for $G$, defined in Definition \ref{def:Cryptographic_PRNG}.
Thus, we have (\ref{eq:main_th_CPRNG}).
\end{proof}

We note that a similar theorem as Theorem \ref{thm:main} holds by assuming the existence of a one-way function, not of a CPRNG.
This is because a CPRNG can be constructed from one-way functions (see, e.g., Ref. \cite{goldreich_2001}, Section 3.4).

\subsection{Generalization of Theorem \ref{thm:main}}
In the previous section, we assumed that QRNG $A$ outputs a bit 0,1 with exactly probability 1/2;
that is, we only consider the case where $A$ realizes the 1-bit TRNG $U^1$.
In addition, the expansion factor $l(k)$ of the CPRNG $G$ that mimics $A^{l(k)}$ could not be chosen freely.
Below we relax these restrictions and generalize Theorem \ref{thm:main}.

\subsubsection{A simple case (biased random bit generators)}
\label{sec:simple_case}

First, we outline the idea with a simple example.

Suppose that QRNG $A$ output 1-bit $X\in\{0,1\}$ as in the previous section, but the probability may be biased; i.e., $p=P(X=1)$ may not be 1/2.
In this case too, essentially the same result as Theorem \ref{thm:main} holds: If there exists a CPRNG $G$, then one can construct a PRNG $G'$ that simulates $A^{l(k)}$, with an arbitrary polynomial $l(k)$ being the expansion factor.
In addition, the outputs of $A^{l(k)}$ and $G'$ are indistinguishable by using any $I$ and $D_I$.

The rigorous proof of this claim will be given in the next section and in Appendix \ref{sec:proof_thm} after generalizing the setting further.
Below we temporarily give a proof sketch for the simple case above (cf. Fig. \ref{fig:biased_simulation}).

First note that $A$ can be efficiently simulated with accuracy $2^{-m}$ in probability, by using a deterministic algorithm $A'$ whose input is the $m$-bit long TRNG $U^m$.
For example, (i) approximate probability $p$ by a fraction $p=b/2^{m}$, $b\in\N$, and (ii) let $A$ output $1$ if the value of $U^m$ is less than or equal to $b$, and $0$ otherwise.

Next let $(A')^n$ denote the $n$ repetition of $A'$,
Clearly, $(A')^n$ is a deterministic algorithm that simulates $A^n$ by using the input of $nm$-bit TRNG $U^{mn}$.

If we then replace the input $U^{mn}$ of $(A')^n$ with the output of CPRNG $G(U^k)$ (having an expansion factor $l(k)=mn$), then we obtain the desired PRNG $G'$ that mimics $A^{n}$.

In this setting, the output of $G'$ is indeed indistinguishable from the output of $A^{l(k)}$, because:
(i) By choosing $m$ sufficiently large, the output $X_{(A')^n}$ of $(A')^n$ can approximate the output $X_{A^n}$ of $A^n$ with arbitrarily small error probability.
That is, $X_{A^n}$ and $X_{(A')^n}$ are indistinguishable.
(ii) The TRNG $U^{mn}$ and the CPRNG $G(U^k)$ are indistinguishable by the definition of CPRNG.
Thus, the outcome of an algorithm $(A')^n$ on their input, namely $X_{(A')^n}$ and $X_{G'}$, are also indistinguishable.

\begin{figure}[t]
  \includegraphics[width=\linewidth]{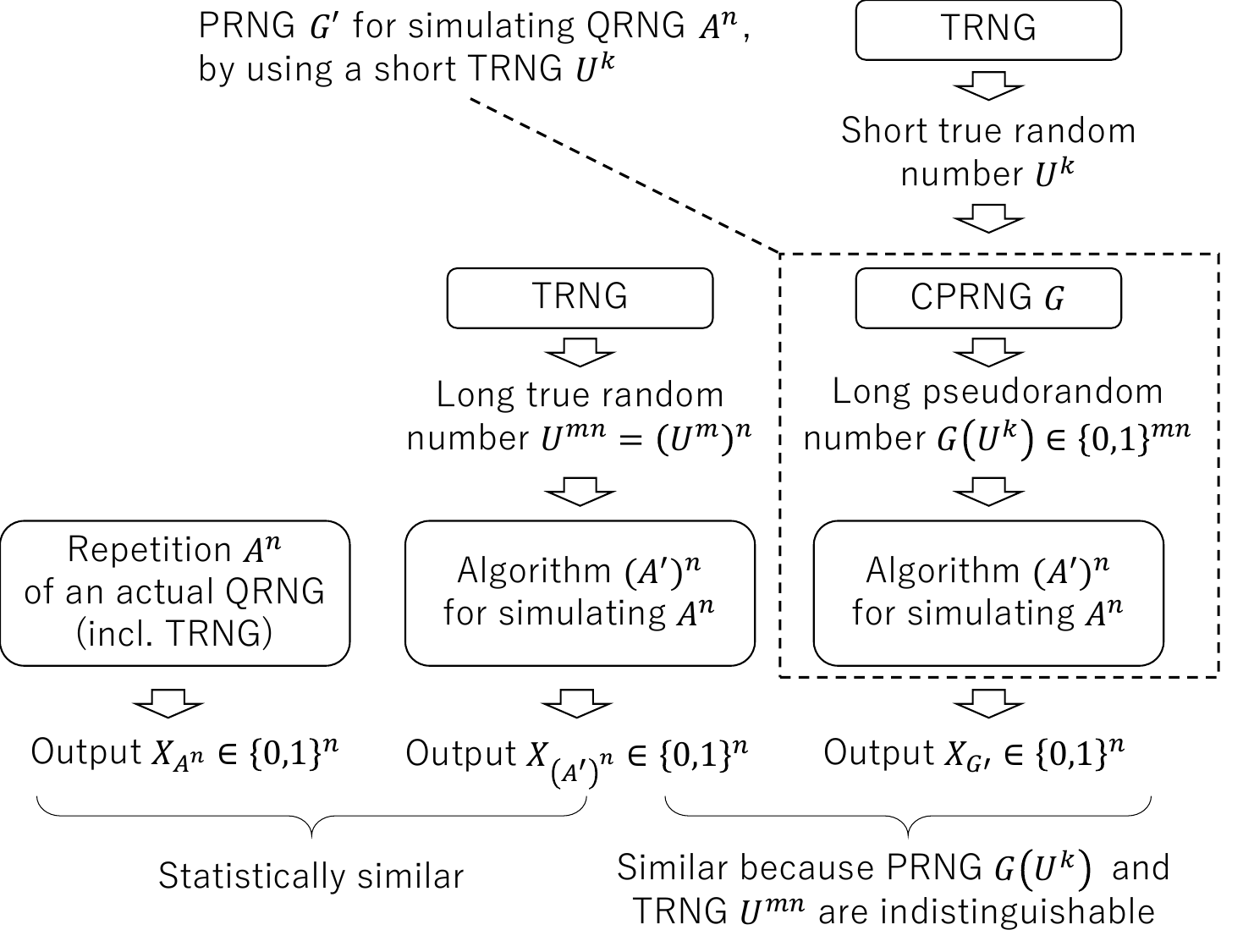}
  \caption{A construction of PRNG $G'$ that simulates the output of QRNG $A$, which is a biased bit.
  }
  \label{fig:biased_simulation}
\end{figure}

\subsubsection{Efficiently classically simulatable systems}

Next, we generalize the settings of the previous sections further and discuss them rigorously.

The generalized settings are as follows.
The output $X$ of a QRNG $A$ is no longer limited to one bit, but can be any finite set ${\cal X}$ (i.e., $X\in{\cal X}$).
Then we require that $A$ be efficiently classically simulatable \cite{doi:10.1098/rspa.2002.1097}.
In recent and fashionable terminology, this means that the system $A$ does not exhibit quantum supremacy.
The precise definition is as follows.
\begin{definition}[Efficient classical simulation] (Ref. \cite{doi:10.1098/rspa.2002.1097}, Def. 2)
\label{def:efficient_simulation}
Suppose that a physical system $A$ outputs a random variable $X\in{\cal X}$.
An efficient classical simulation for $A$ is an efficient probabilistic algorithm $A'$ which, on input $1^n$ (sequence of $n$ symbols of `1'), outputs $\in{\cal X}$ that satisfies 
\begin{align}
\left\|P(X_A)-P(X_{A'(1^n)})\right\|_1\le2^{-n}.
\label{eq:efficient_simulation}
\end{align}
We say that $A$ is efficiently classically simulatable (ECS) if such $A'$ exists.
\end{definition}

Several technical remarks are in order related with this definition:
First, the symbol $\|P(X)-P(Y)\|_1$ denotes the total variational distance between the probability distributions of random variables $X,Y\in{\cal X}$;
$\|P(X)-P(Y)\|_1:=\frac12\sum_{x\in{\cal X}}\left|P(X=x)-P(Y=x)\right|$.
Second, the input $1^n$ is meant to specify parameter $n$ and to require that $A'$ halts within a polynomial time of $n$.
Third, in the original definition (i.e., in Definition 2 of Ref. \cite{doi:10.1098/rspa.2002.1097}) $A'$ has an additional input besides $1^n$, but we here omit it because in this paper we restrict ourselves to the case where physical system $A$ is a QRNG and thus has no input.
Finally, in the statement of Definition \ref{def:efficient_simulation}, we used the term `physical system' instead of `quantum system' because the discussions below hold not only for quantum systems, but for physical systems in general.

Additionally, to comply with the fact that the output of $A$ is a ${\cal X}$ string rather than a bit string, we let the input to index $I$ also be an ${\cal X}$ string. That is, we let $I$ be an efficient probabilistic algorithm $I:{\cal X}^*\to\{0,1\}^*$.

Under these settings, we have the following generalization of Theorem 1.
\begin{theorem}
\label{thm:main_gen}
Suppose that there exists a CPRNG.
Then for any ECS physical system $A$, and for any positive and increasing polynomial $l(k)$, there exists a deterministic algorithm $B:\{0,1\}^*\to {\cal X}^*$ which, on input a $k$-bit string, outputs a ${\cal X}$ string of length $l(k)$ that satisfy 
\begin{align}
&\left|P(D_I(I(B(U^k)))=1)-P(D_I(I(X_{A^{l(k)}}))=1)\right|\nonumber\\
&\le{\rm negl}(k).
\label{eq:thm_equation}
\end{align}
for any pair of a randomness measure $I$ and a distinguisher $D_I$.
\end{theorem}
In other words, for any ECS system $A$, there exists an algorithm $B$ that impersonates it with low randomness.
Moreover, the output size $l(k)$ of $B(U^k)$ can be chosen to be an arbitrary polynomial.

The proof of Theorem \ref{thm:main_gen} is given in Appendix \ref{sec:proof_thm}.

%%%%%%%%%%%%%%%%%%%%%%%%%%%%%%
\section{Empirical analyses}
\label{supports}
%%%%%%%%%%%%%%%%%%%%%%%%%%%%%%

%\textcolor{blue}{
In this section, we give empirical examples to support Theorems \ref{thm:main} and \ref{thm:main_gen}.
For this, we perform
%}
%In this paper, we deduce our conclusions mainly from 
the data analyses on two kinds of QRNGs: quantum coin toss games and the Bell tests.
We start this section by presenting the details of those datasets.

%%%%%%%%%%%%%%%%%%%%%%%%%%%%%%
\subsection{Data}
\label{material}
\subsubsection{Quantum coin toss games}
%\label{quantum}
%%%%%%%%%%%%%%%%%%%%%%%%%%%%%%

%We performed quantum coin toss games by using IBMQ and its QASM simulator.
In quantum coin toss games, a type of QRNGs, we first prepare the initial state $\ket{\Psi}=(\ket{0}+\ket{1})/\sqrt{2}$, where $\{\ket{0}, \ket{1}\}$ is an orthonormal basis of the state space $\C^2$ of a two-level system.
We then perform projection measurement on the initial state $\ket{\Psi}$ along the computational bases $\{\ket{0}, \ket{1}\}$ and find the outcome $a\in\{0,1\}$.
Repeating this process $N$ times, we obtain a bit string $a^N=(a_1, a_2,\dots, a_N)$, where $a_i\in\{0,1\}$.

%%%%%%%%%%
\begin{table}[b]
%\begin{table}[htpb]
\caption{Mean relative frequency of bit value 1 and standard deviation of three datasets of 100 bit strings with length $N=20000$ created by the quantum coin tosses. The numbers after $\pm$ are the standard deviation obtained after 100 trials.}
\begin{center}
\begin{tabular}{l|c}
\hline
\multirow{2}{*}{Device used in data acquisition} & Mean relative frequency\\
 & of bit value 1 \\
\hline\hline
\verb|ibmq_manila| & $0.4851\pm0.0034$\\
\verb|ibmq_qasm_simulator| & $0.4999\pm0.0034$\\
\verb|ibmq_qasm_simulator| (with noise) & $0.4893\pm0.0036$\\
%(with noise model) & \\
\hline
\end{tabular}
\end{center}
\label{mean}
\end{table}%
%%%%%%%%%%

In IBM Quantum, the initial state is restricted to the computational basis state $\ket{0}$. The state $\ket{\Psi}$ is implemented as the application of the Hadamard gate 
$H=(\ket{0}\bra{0}+\ket{0}\bra{1}+\ket{1}\bra{0}-\ket{1}\bra{1})/\sqrt{2}$ to the computational basis state $\ket{0}$.
Therefore, we can rewrite the quantum coin toss as the following three steps: i) state preparation in the computational basis $\ket{0}$ for a two-level system (qubit) in a quantum processing unit (QPU), ii) application of the Hadamard gate $H$, iii) projection measurement along the computational bases $\{\ket{0}, \ket{1}\}$.

Following the above steps, we made a dataset by using \verb|ibmq_manila| provided by IBM Quantum.
%, \verb|ibmq_qasm_simulator|, and \verb|ibmq_qasm_simulator| with the device noise model of \verb|ibmq_manila|.
The dataset was made on February 23, 2023, and consists of 100 bit strings with the length $N$ aligned to $N=20000$.
%Note that the datasets created from the latter two devices are regarded as sets of pseudo-random numbers since the simulators run on classical computers, which may perform the simulation using short seed random numbers.

To make a comparison between this QRNG and PRNGs, we performed classical simulations using three types of PRNGs. 
The first and second types are \verb|ibmq_qasm_simulator|, and \verb|imbq_qasm_simulator| with the device noise model of \verb|ibmq_manila|, respectively. 
The third is ChaCha20 algorithm described in Refs. \cite{rfc8439,books/crc/KatzLindell2020}.
Note that the simulators are regarded as PRNGs since they run on classical computers and use a short random seed.
We hereafter specify datasets by the name of the device used for their data acquisition process.

Table~\ref{mean} shows the mean relative frequency of bit value 1 in the entire bit string.
The numbers after $\pm$ are the standard deviation obtained after 100 trials.
Note that the mean relative frequency of \verb|ibmq_qasm_simulator| is closer to the theoretical value $1/2$ than that of \verb|ibmq_manila|, and these two values are mutually distinguishable within the standard deviations.
Note also that, in contrast, the mean relative frequency of \verb|ibmq_qasm_simulator| (with noise) is comparable to that of \verb|ibmq_manila| within the standard deviations.

%%%%%%%%%%%%%%%%%%%%%%%%%%%%%%
\subsubsection{Innsbruck experiment}
%\label{quantum}
%%%%%%%%%%%%%%%%%%%%%%%%%%%%%%

%We hereafter reexamine the hypothesis that the degree of the violation of the Bell inequality influences the randomness of the bit strings obtained from the Bell test.
To examine observation ii) in the Introduction, as an example of the actual Bell tests, we employ the dataset taken for \cite{PhysRevLett.81.5039}, which reports the first experiment of the Bell test free from the locality loophole.
We hereafter call the series of the Bell tests conducted in \cite{PhysRevLett.81.5039} the Innsbruck experiments.

In the Innsbruck experiments, entangled photon pair $\ket{\Phi}=(\ket{0}\ket{1}-\ket{1}\ket{0})/\sqrt{2}$ is created.
%Here $\{\ket{0}, \ket{1}\}$ is the orthonormal basis of the one-particle state space $\C^2$ of the photon polarization.
Each photon in the pair is distributed to space-like separated two points (hereafter called Alice and Bob, respectively), but some photons are lost during the distribution process.
For the distributed photons in each distribution process, Alice performs a measurement by choosing the measurement axis $x$ out of two alternatives $x\in\{0,1\}$, and obtains the outcome $a\in\{0,1\}$.
Bob also does so for each photon.
We hereafter denote Bob\rq s measurement axis and outcome by $y\in\{0,1\}$ and $b\in\{0,1\}$, respectively.

Repeating the above measurement processes many times, they store the data of the measurement outcomes with their time stamps and the records of the measurement axes.
We call each of these records of the series of measurements a sample.
In particular, we shall analyze 21 samples, each named longdist*, which are the samples actually used in the Innsbruck experiments \cite{weihs_gregor_1998_7185335}.

Given a sample, we extract a pair of {\it coincident outcomes}, namely, those detection events that occur simultaneously on both Alice\rq s and Bob\rq s sides within the threshold given in \cite{weihs_gregor_1998_7185335} with the time stamps\rq~difference. 
We denote the outcomes by $(a,b)$, and the number of the pairs of the coincident outcomes in the sample by $N$.
We construct three bit strings, that is, Alice\rq s bit string
$
a^N=(a_1,a_2,\dots,a_N),
$
Bob\rq s bit string
$
b^N=(b_1,b_2,\dots,b_N),
$
and a mixed bit string
$
c^N=(a_1, b_1, a_2, b_2,\dots, a_N, b_N).
$
%Note that we can also construct the bit strings from the measurement axes:
%Alice\rq s bit string
%$
%x^N=(x_1,x_2,\dots,x_N),
%$
%Bob\rq s bit string
%$
%y^N=(y_1,y_2,\dots,y_N),
%$
%and a mixed bit string
%$
%z^N=(x_1, y_1, x_2, y_2,\dots, x_N, y_N).
%$

%Moreover, we can construct shorter substrings of $a^N$ and $b^N$ conditioned by the measurement axes such as ordered sets $\alpha^N_{(k)}=\{a_i\,|\, x_i=k\}$ and $\beta^N_{(k)}=\{b_i\,|\, y_i=k\}$, where $x_i$ and $y_i$ are the measurement axes for $a_i$ and $b_i$, respectively.  $a_i$ in $\alpha^N_{(k)}$ and $b_i$ in $\beta^N_{(k)}$ obey the increasing order of $i$.

By construction, we may expect that the bit string $c^N$ is less random than $a^N$ and $b^N$, if the given sample violates the Bell inequality.
For, every pair $(a_i,b_i)$ in $c^N$ has a strong correlation not described by any local hidden variable theories (LHVTs), and thereby $c^N$ becomes more structured than the bit strings expected from LHVTs.
This inherent difference between $c^N$ and the others is also examined in this paper.

%%%%%%%%%%%%%%%%%%%%%%%%%%%%%%
\subsubsection{Efficient Simulation using ChaCha20}
%\label{quantum}
%%%%%%%%%%%%%%%%%%%%%%%%%%%%%%
%\textcolor{blue}{
To give numerical support to Theorems \ref{thm:main} and \ref{thm:main_gen},  by using the PRNG called ChaCha20 algorithm \cite{rfc8439, books/crc/KatzLindell2020}, we made the pseudo-random number strings that mimic the bit strings obtained from the above two experiments by following the procedure given in Sec.~\ref{sec:simple_case} for the biased random bit sequences.

More precisely, for the Innsbruck experiment, we used  the following process.
First given the relative frequency and length of an experimental bit string, we generated a pseudo-random number byte string having twice the length of the given bit string.
Then we convert each two bytes of the pseudo-random number byte string to a real number.
We compare each of these real numbers with a threshold value, which is the relative frequency value rescaled by $2^{16}$, and convert it to a bit 0 (1) if it is higher (lower) than the threshold.

We also performed the same procedure for the quantum coin toss games by using the mean relative frequency of the data taken by \verb|ibmq_manila|.
%} 

%%%%%%%%%%%%%%%%%%%%%%%%%%%%%%
\subsection{Measures}
\label{measures}
%%%%%%%%%%%%%%%%%%%%%%%%%%%%%%

We characterize the bit strings in the datasets described in the previous section by the (algorithmic) randomness. 
In addition to this, for the dataset of the Innsbruck experiment, we measure how much the data are non-local.
We hereafter introduce the measures of algorithmic randomness and non-locality.

%%%%%%%%%%%%%%%%%%%%%%%%%%%%%%
\subsubsection{Randomness measures}
%\label{feasibility}
%%%%%%%%%%%%%%%%%%%%%%%%%%%%%%

%%%%%%%%%%%%%%%%%%%%%%%%%%%%%%
%\subsubsection{Relative frequencies}
%\label{feasibility}
%%%%%%%%%%%%%%%%%%%%%%%%%%%%%%

%%%%%%%%%%%%%%%%%%%%%%%%%%%%%%
\paragraph{Lempel-Ziv (LZ) complexity}
\label{sec:LZ_complexity}
%%%%%%%%%%%%%%%%%%%%%%%%%%%%%%

%The LZ complexity is based on the LZ76 data compression algorithm \cite{1055501}.
Let us suppose that we are given a bit string $s^N=(s_1,s_2,\dots s_N)\in\{0,1\}^N$, and consider the following parsing 
\begin{equation}
E(s^N)=s(1,h_1)s(h_1+1,h_2)\dots s(h_{m-1}+1,N),
\end{equation}
where $s(i,j)=(s_i, s_{i+1},\dots,s_j)$ is the substring of $s^N$ from the $i$-th bit to the $j$-th.
Now let us set $h_1=1$ and make the parsing in such a way that $s(h_{i-1}+1,h_i)$ is the minimal substring that has not already appeared as the substring $s(h_{k-1}+1,h_k)$ with $k<j$.
Then, it has been known that the $E(s^N)$ is uniquely determined if we follow the above rule of parsing \cite{1055501}.

Let $c(N)$ be the number of the substrings constructed in the above rule.
Then $c(N)$ has an upper bound \cite{1055501}
\begin{eqnarray}
c(N)&<&\frac{N}{(1-\varepsilon_N)\log_2(N)},
\label{eq:bound_on_c(S)}
\end{eqnarray}
where
\begin{eqnarray}
\varepsilon_N&=&2\frac{1+\log_2\log_2(2N)}{\log_2(N)}.
\end{eqnarray}
Note that Ineq.~(\ref{eq:bound_on_c(S)}) follows solely from the definition of $c(N)$, and thereby holds for any bit string $s^N$ of length $N$.

Motivated by this upper bound of $c(N)$, we define
\begin{equation}
K(N)=c(N){\log_2N\over N},
%\label{eq:bound_on_K(N)}
\end{equation}
which has been widely used as a measure of algorithmic complexity in literature \cite{PhysRevA.36.842,10.1063/1.4808251,PhysRevA.82.022102,Solis_2015,e20110886,PhysRevA.98.042131,Abbott_2019}.
We hereafter call $K(N)$ the LZ complexity in this paper, although much literature calls $K(N)$ the normalized LZ complexity.

As suggested in Ineq.~(\ref{eq:bound_on_c(S)}), the upper bound of the LZ complexity $K(N)$ depends on the length $N$.
To make it reasonable that the comparison of the algorithmic complexity between the bit strings with mutually distinct lengths, let us normalize the LZ complexity $K(N)$ by the maximal value of $c(N)$ for the given $N$.
For this, we recall that the maximal value of $K(N)$ is attained by the bit string whose substrings are exhaustively aligned and lengthened from left to right, for example,
\begin{equation}
t^N=(0,1,0,0,0,1,1,0,1,1,\ldots),
%\label{eq:bound_on_K(N)}
\end{equation}
whose parsing is given as
\begin{eqnarray}
t(1,1)&=&(0),\nonumber\\
t(2,2)&=&(1),\nonumber\\
t(3,4)&=&(0,0),\nonumber\\
t(5,6)&=&(0,1),
\end{eqnarray}
and so on.
In this case, the length of the bit string composed of all the mutually distinct bit strings whose lengths are up to $m$ is given by
\begin{equation}
l_m=\sum_{k=1}^mk2^k=(m-1)2^{m+1}+2.
%\label{eq:bound_on_K(N)}
\end{equation}
Let $m^*$ be the number to minimize the difference $N-l_m\ge0$.
Then we obtain the number of  the bit strings to compose $t^N$ as
\begin{equation}
c_{\rm max}(N)=\sum_{k=1}^{m^*}2^k+\left\lceil{N-l_{m^*} \over m^*+1}\right\rceil=2^{m^*+1}-2+\left\lceil{N-l_{m^*} \over m^*+1}\right\rceil,
%\label{eq:bound_on_K(N)}
\end{equation}
where $\lceil x\rceil$ is the minimal integer exceeding $x$.
It follows that the maximum of $K(N)$ is given as
\begin{equation}
K_{\rm max}(N)=c_{\rm max}(N){\log_2N\over N}.
%\label{eq:bound_on_K(N)}
\end{equation}
Based on the above argument, we define the normalized LZ complexity
\begin{equation}
\kappa(N)={K(N) \over K_{\rm max}(N)}={c(N) \over c_{\rm max}(N)}.
%\label{eq:bound_on_K(N)}
\end{equation}

%The asymptotic upper bound of $K(N)$ is given by
%\begin{equation}
%K(N)\le h,
%\label{eq:bound_on_K(N)}
%\end{equation}
%where $h$ is the entropy rate defined  by
%\begin{equation}
%h=\lim_{N\rightarrow\infty}{H(N)\over N},
%\label{eq:bound_on_K(N)}
%\end{equation}
%with $H(N)$ being the block entropy \cite{1055501}.
%Here the block entropy is {\bf Blah Blah Blah.}

%\begin{enumerate}
%\item エントロピー（レート）との関係を（必要なら）説明
%\end{enumerate}

To close this introduction of the LZ complexity, let us point out the inconsistencies appearing in data analyses of Ref. \cite{PhysRevA.98.042131}.
%For this, we recall Theorem 2 of Ref. \cite{1055501} in the notation of Kovalsky et al. \cite{PhysRevA.98.042131}:
%The LZ complexity $c(N)$ of any bit string $s^N=(s_1,s_2,\dots,s_N)\in\{0,1\}^N$ of length $N$ satisfies
%\begin{eqnarray}
%c(N)&<&\frac{N}{(1-\varepsilon_N)\log_2(N)},
%\label{eq:bound_on_c(S)}
%\end{eqnarray}
%where
%\begin{eqnarray}
%\varepsilon_N&=&2\frac{1+\log_2\log_2(2N)}{\log_2(N)}.
%\end{eqnarray}
%Note that Ineq.~(\ref{eq:bound_on_c(S)}) follows solely from the definition of the LZ complexity, and thereby holds for any bit string $S$ of length $N$.
Using the fact that $\varepsilon_N$ decreases monotonically with $N$ and that $\varepsilon_{4096}<0.784$ holds, we obtain
\begin{equation}
c(N)<4.6\frac{N}{\log_2(N)}\quad {\rm for}\quad 4,096\le N.
\end{equation}
Thus for the normalized complexity measure $K(N)$,% defined in eq. (1) of Ref. \cite{PhysRevA.98.042131},
\begin{equation}
K(N)< 4.6\quad {\rm for}\quad 4,096\le N
%K(N)=c(N)\times \log_2(N)/N< 4.6\quad {\rm for}\quad 4,096\le N
\label{eq:bound_on_K(N)}
\end{equation}
must always hold.

In contrast, almost half of all trials shown in Table I of Ref. \cite{PhysRevA.98.042131} violate the upper bound (\ref{eq:bound_on_K(N)}).
More precisely, for 16 of 37 trials in Table I (namely, longdist$X$ with $X=1,\dots,4$, 10, 11, 13, 20, 23, 30, 32, 33, 34, 36 and 37, and Conlt3), the authors conclude the $K(N)>7$ and $N>4,096$ hold simultaneously, but this clearly conflicts with Ineq.~(\ref{eq:bound_on_K(N)}), implying inaccuracy in their data processing.

%%%%%%%%%%%%%%%%%%%%%%%%%%%%%%
\paragraph{Borel normality measure}
\label{sec:Borel_normality}
%%%%%%%%%%%%%%%%%%%%%%%%%%%%%%
The Borel normality measures the distribution of substrings in the given bit string  \cite{calude2002}.
Let $B_m=\{0,1\}^m$ be the set of the bit strings with length $m$, and $N_j^m(s^N)$ be the number of occurring the lexicographical $j$-th bit string of length $m$ in $s^N$.
Let us denote the length of $s^N$ over $B_m$ by $|s^N|_m$ and define $|s^N|=|s^N|_1$. 
Then the bit string $s^N$ is Borel normal (with accuracy ${1\over\log_2|s^N|}$) if we have
\begin{eqnarray}
\left|{N_j^m(s^N) \over |s^N|_m}-{1\over 2^m}\right|\le{1\over\log_2|s^N|},
%\label{eq:bound_on_c(S)}
\end{eqnarray}
for every integer $1\le m\le\lfloor\log_2\log_2|s^N| \rfloor$ and $1\le j\le 2^m$ \cite{calude2002}.

From the above definition of the Borel normality, the following measure was proposed in \cite{Abbott_2019}:
\begin{eqnarray}
B(s^N)=\max\left|{N_j^m(s^N) \over |s^N|_m}-{1\over 2^m}\right|\log_2|s^N|,
%\label{eq:bound_on_c(S)}
\end{eqnarray}
where the maximum is over $m=1, 2, \dots, \lfloor\log_2\log_2|s^N| \rfloor$ and $j=1,2,\dots, 2^m$.
The measure $B(s^N)$ ensures that the bit string $s^N$ is Borel normal if $B(s^N)\le1$.

%%%%%%%%%%%%%%%%%%%%%%%%%%%%%%
\subsubsection{Non-locality measure}
%\label{feasibility}
%%%%%%%%%%%%%%%%%%%%%%%%%%%%%%
We first introduce the CHSH correlation function \cite{PhysRevLett.23.880}
\begin{eqnarray}
S=\left|\sum_{x,y}(-1)^{xy}\[P(A=B|xy)-P(A \neq B|xy)\]\right|,
%\label{eq:bound_on_c(S)}
\end{eqnarray}
where $P(A=B|xy)$ is the conditional probability of the outcome coincidence given the pair of the measurement axes $(x,y)$ and $P(A\neq B|xy)$ is that of the non-coincident outcomes. 
The CHSH correlation function has the upper bound $S\le2$ for any LHVTs  \cite{PhysRevLett.23.880} whereas for quantum mechanics it  exceeds the bound and attains $2\sqrt{2}\approx2.828$ \cite{Cirel'son1980}.
We employ the CHSH correlation function as the measure of non-locality.

%%%%%%%%%%%%%%%%%%%%%%%%%%%%%%
\subsection{Results of the empirical analyses}
\label{result}
%%%%%%%%%%%%%%%%%%%%%%%%%%%%%%

%%%%%%%%%%%%%%%%%%%%%%%%%%%%%%
\subsubsection{Quantum coin tosses}
%\label{feasibility}
%%%%%%%%%%%%%%%%%%%%%%%%%%%%%%

%%%%%%%%%%
\begin{figure}[t]
%\begin{minipage}[a]{.7\linewidth}
  \includegraphics[width=.4\textwidth]{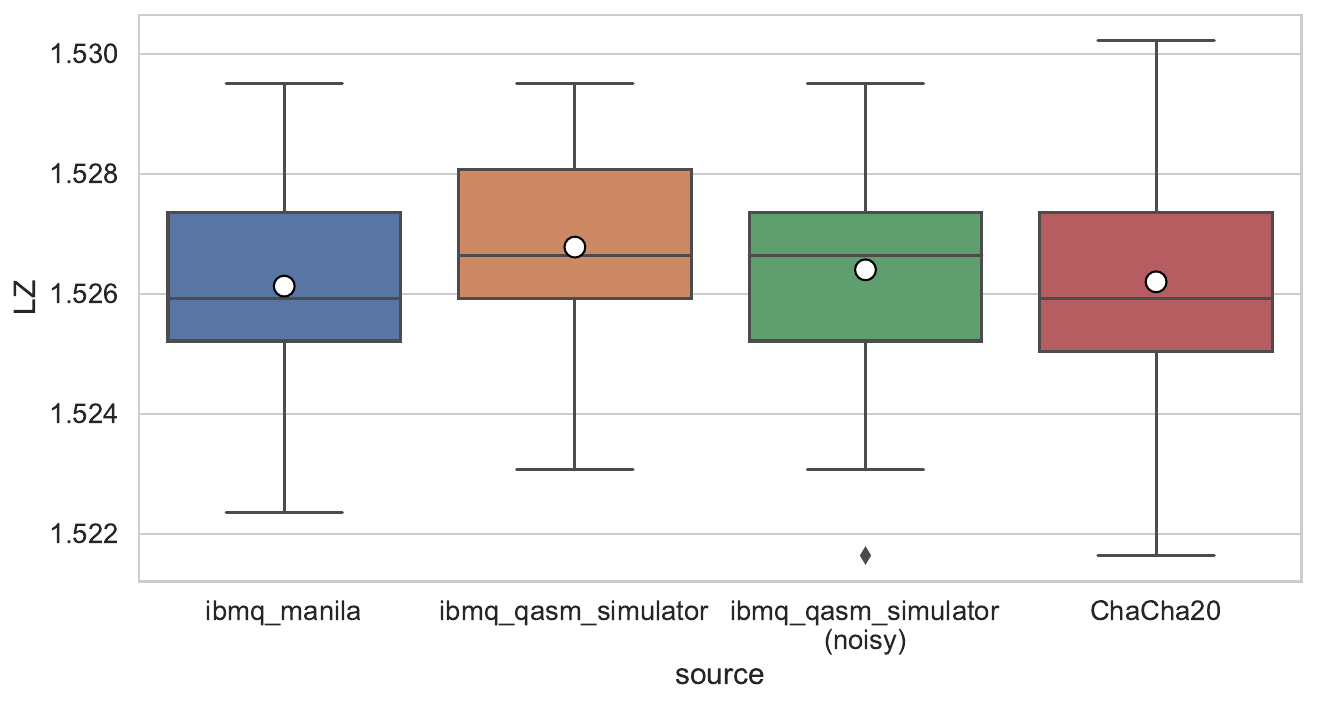}
 % \subcaption{LZ complexity.}
 % \label{qcoinLZ}
 % \end{minipage}
%\begin{minipage}[a]{.7\linewidth}
  \includegraphics[width=.4\textwidth]{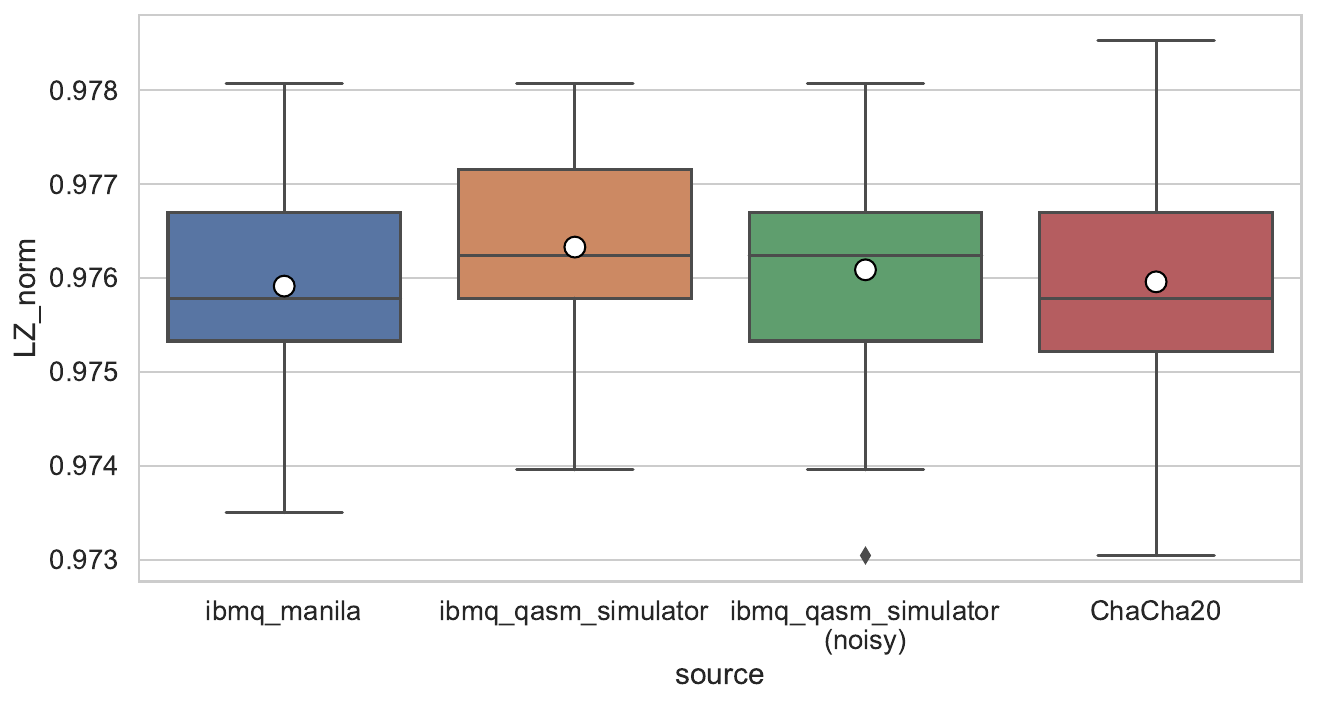}
  %\subcaption{Normalized LZ complexity.}
  %\label{qcoinLZnorm}
  %\end{minipage}
%\begin{minipage}[a]{.7\linewidth}
  \includegraphics[width=.4\textwidth]{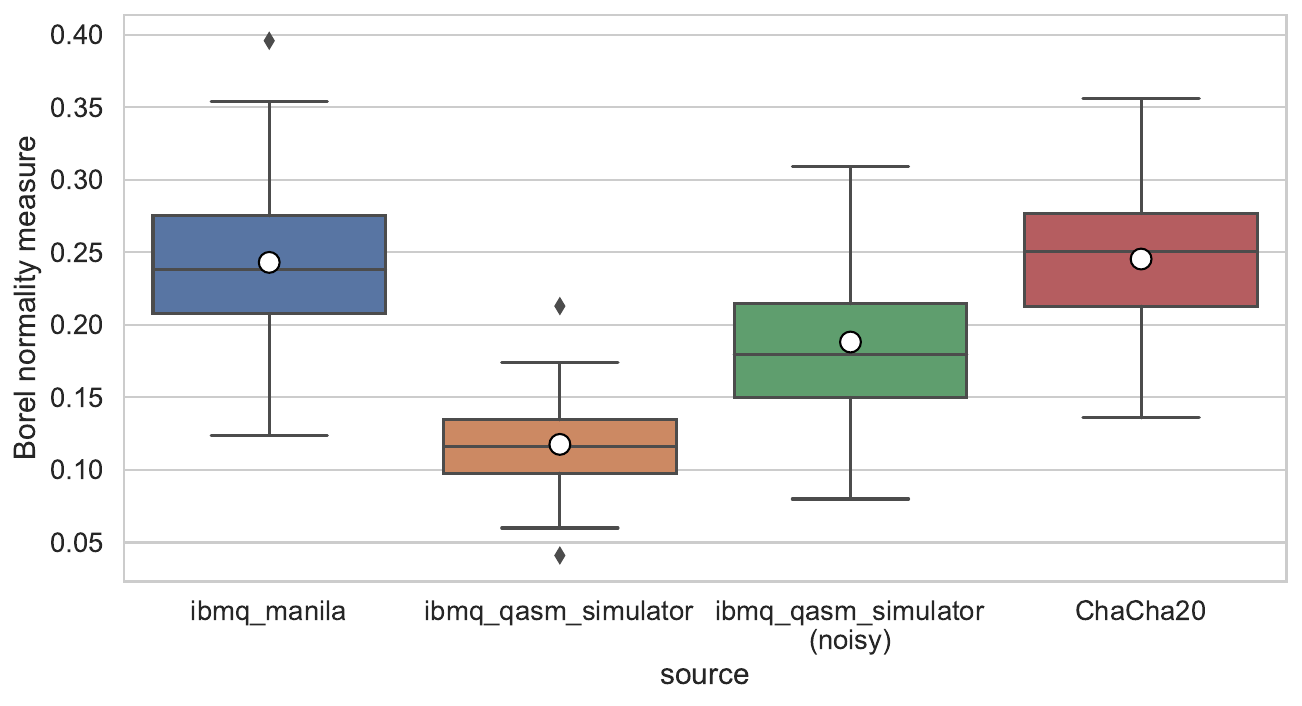}
  %\subcaption{Borel normality measure.}
  %\label{qcoinBorel}
  %\end{minipage}
%(a)  \includegraphics[width=0.4\textwidth]{q2_conditional_flat_postprocess1_clifford_deviation.pdf}
%\\
%(b)    \includegraphics[width=0.4\textwidth]{q2_conditional_flat_postprocess2_clifford_deviation.pdf}
  \caption{The box plots of the randomness measures for the quantum coin tosses: (Above) the LZ complexity, (Middle) the normalized LZ complexity, and (Below) the Borel normality measure.
Throughout this paper, we draw the box-and-whisker plots as follows.
%In box-and-whisker plots, 
Given the data and their attributes to be plotted, the length of the upper (lower) whisker is set as a maximum (minimum) value not exceeding 1.5 times the interquartile range (IQR, the height of the box). 
Data whose attribute values are outside this range are called outliers, which are plotted as dots.
%Throughout this paper, we do not plot the outliers in the box-and-whisker plots, in order to see the typical behaviors of the data in detail.
  }
  \label{box_qcoin}
\end{figure}
%%%%%%%%%%

We discuss whether there is a difference in the randomness of the bit strings created by a QRNG (quantum computer) and PRNGs including simulators by considering the simplest setup, {\it i.e.}, quantum coin toss games.
Figure \ref{box_qcoin} shows the randomness measures (the algorithmic complexity measures mentioned above) of our datasets.

As to the (normalized) LZ complexity, the bit strings obtained from the simulator with the device noise model and ChaCha20 are comparable to those obtained from the QRNG.
The bit strings from the simulator (without the device noise model) have a slightly higher algorithmic complexity than those from the QRNG.
This is presumably because the mean relative frequency of the data obtained from the simulator is almost $0.5$, which is significantly higher than that of the dataset from the actual device.
% in general, irrespective of the use of the noise model.

In contrast, the Borel normality measure distinguishes all the cases.

Let us look at the above observation by using the hypothesis testing (Welch\rq s $t$-test) on whether the simulations have no difference from the actual device:
For the case of noiseless simulation against the actual device, we found the statistic $t\approx-3.01$ with its $p$-value being $2.96\times10^{-3}$ for the LZ complexity and normalized one, whereas $t\approx20.69$ with its $p$-value being $2.15\times10^{-51}$ for the Borel normality measure.
This suggests that these two ways of creating the random bit strings are distinguishable from each other if we focus on the average values of the randomness measures.

In contrast, for the case of simulation with the device noise model against the actual device, we observed the statistic $t\approx-1.25$ with its $p$-value being $0.213$ for the LZ complexity and normalized one, while $t\approx7.58$ with its $p$-value being $1.33\times10^{-12}$ for the Borel normality measure.
This shows that these two ways of creating the random bit strings are indistinguishable from each other if we employ the (normalized) LZ complexity as shown in Theorem \ref{thm:main}, if their relative frequency is comparable.
Moreover, for the case of PRNG, the statistic is $t\approx-0.32$ and its $p$-value is $0.749$ for the LZ complexity and normalized one, while $t\approx-0.33$ with its $p$-value being $0.739$ for the Borel normality measure.
Thus, the bit strings created by ChaCha20 are indistinguishable for all the complexity measures we have employed, supporting our theorems.

%%%%%%%%%%
\begin{figure}[t]
%\begin{minipage}[a]{.8\linewidth}
  \includegraphics[width=.5\textwidth]{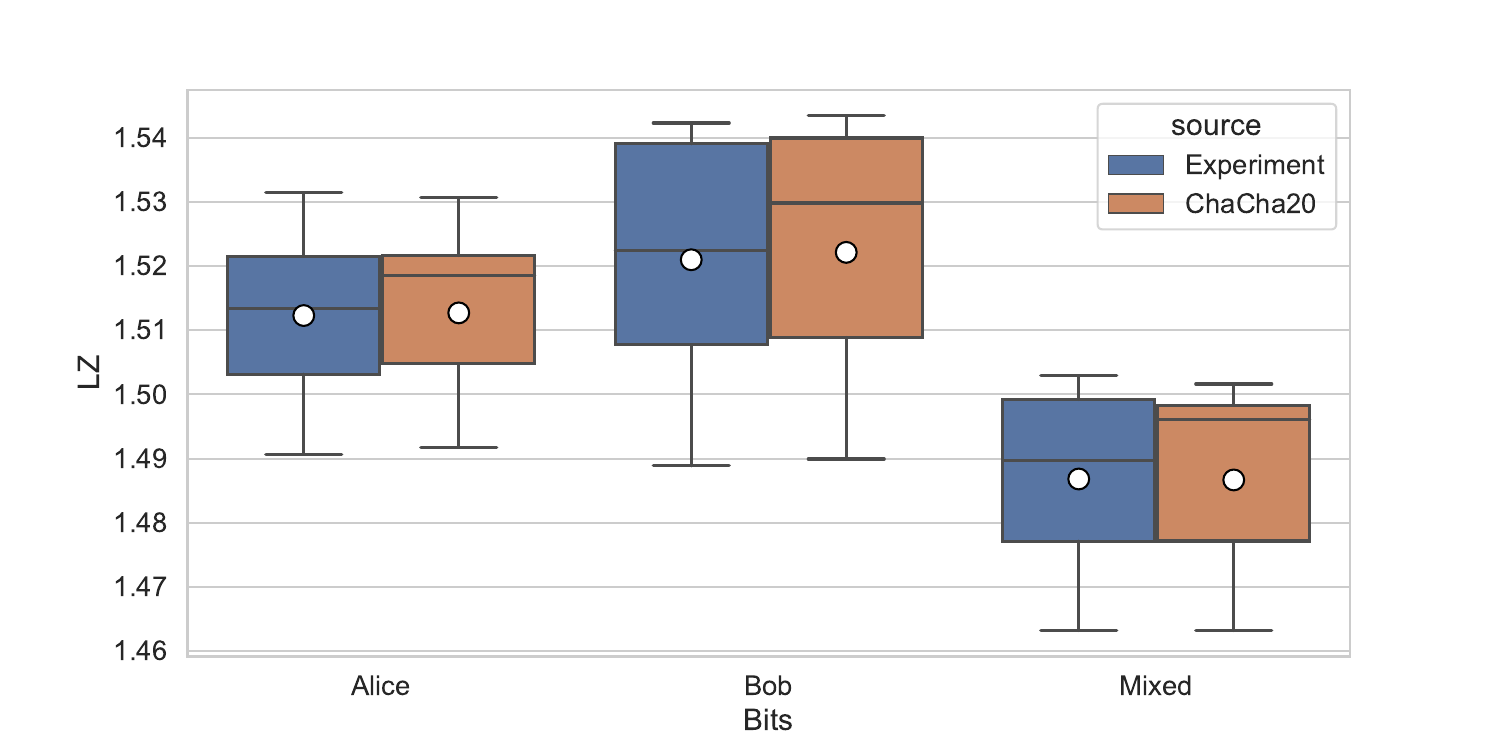}
 % \subcaption{LZ complexity.}
 % \label{box_LZ}
 % \end{minipage}
%\begin{minipage}[a]{.8\linewidth}
  \includegraphics[width=.5\textwidth]{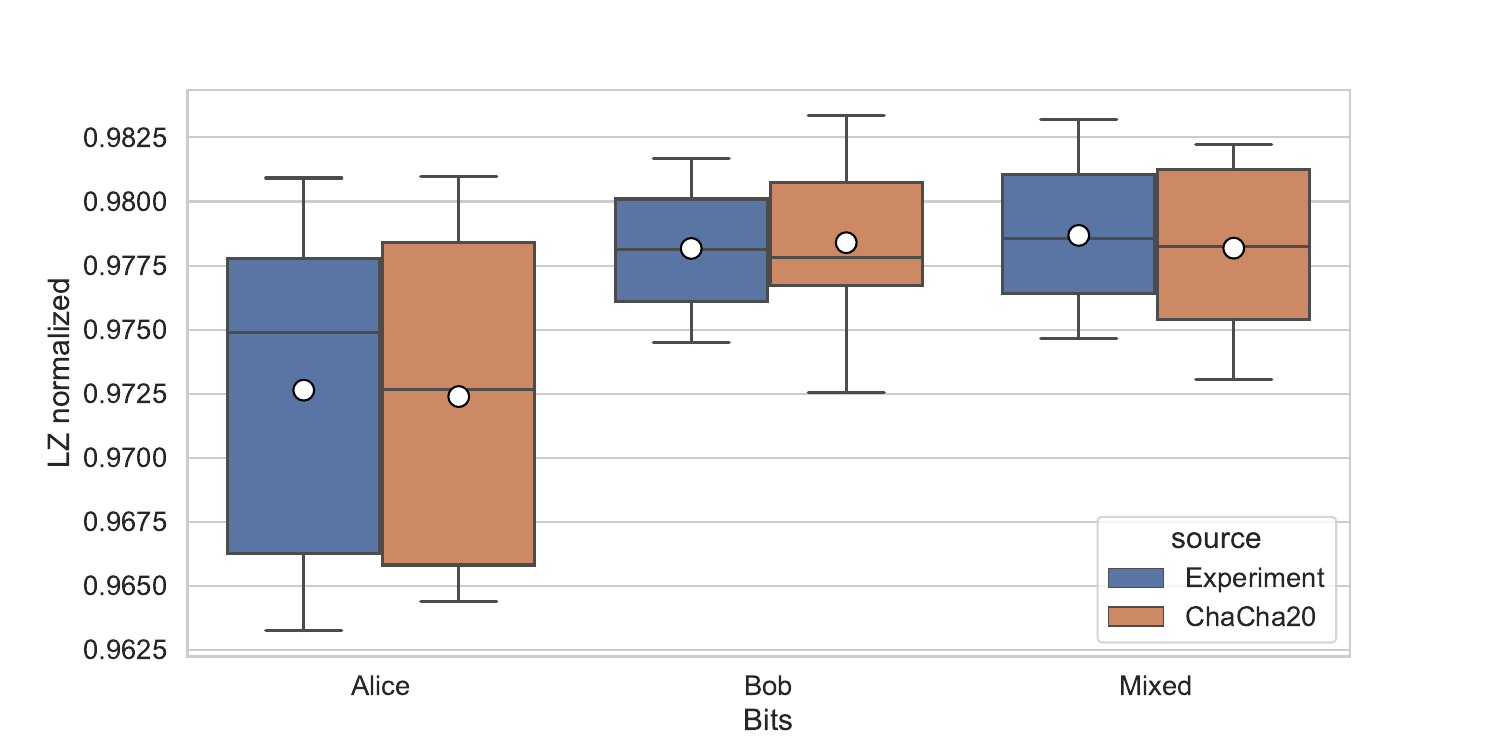}
 % \subcaption{Normalized LZ complexity.}
%  \label{box_LZnorm}
 % \end{minipage}
%\begin{minipage}[a]{.8\linewidth}
  \includegraphics[width=.5\textwidth]{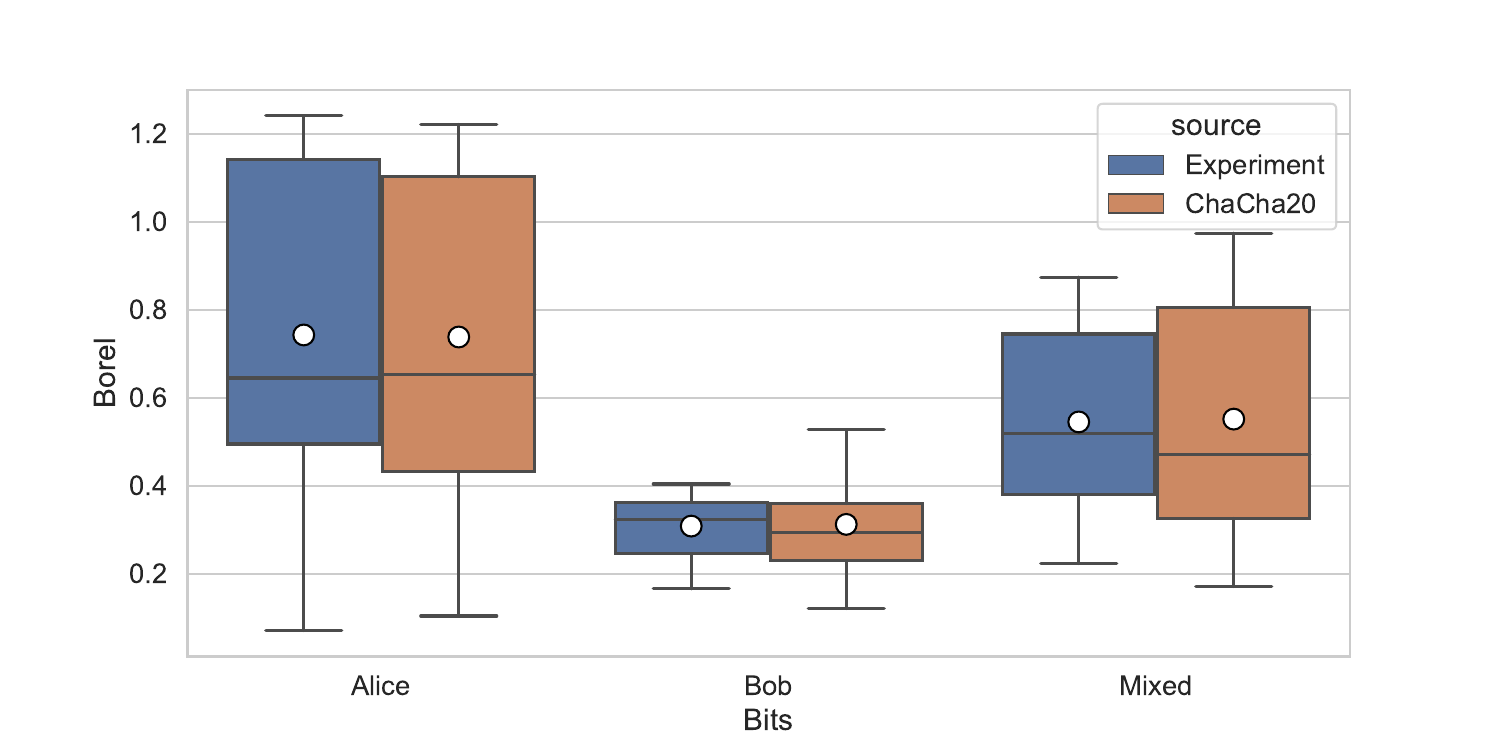}
 % \subcaption{Borel normality.}
 % \label{box_LZnorm}
 % \end{minipage}
%(a)  \includegraphics[width=0.4\textwidth]{q2_conditional_flat_postprocess1_clifford_deviation.pdf}
%\\
%(b)    \includegraphics[width=0.4\textwidth]{q2_conditional_flat_postprocess2_clifford_deviation.pdf}
  \caption{The box plots of the randomness measures for the bit strings generated from the Innsbruck experiment and ChaCha20: (Above) the LZ complexity, (Middle) the normalized LZ complexity, and (Below) the Borel normality measure. The outlier coming from longdist4 is omitted. The bit strings at the Bob\rq s station are more random than those at the Alice\rq s station. The mixed bit strings are less random if evaluated by the LZ measure.
  }
  \label{boxes}
\end{figure}
%%%%%%%%%%

%%%%%%%%%%%%%%%%%%%%%%%%%%%%%%
\subsubsection{Innsbruck experiment}
%\label{feasibility}
%%%%%%%%%%%%%%%%%%%%%%%%%%%%%%

Motivated by our theorems and the discussion presented in Sec.~\ref{measures},
we proceed to the analyses of the Innsbruck experiment, or the longdist dataset. %Motivated by the discussion in Sec.~\ref{measures}, we analyzed the longdist dataset.
%In this subsection, we present the results of these empirical analyses.
More concretely, we shall show the descriptive statistics of the randomness measures mentioned above and their correlation coefficients.

In our analyses, we have omitted the longdist4 sample as an outlier due to its relatively small number of the coincident events $N$; see Appendix~\ref{measure_data}.

%By using these statistics, we perform the test of no correlation between the non-locality measure and the randomness measures.

%Before we proceed to present our empirical results, we mention the outlier data in the longdist dataset.

Figure \ref{boxes} shows the descriptive statistics of the randomness measures.
All the randomness measures show that Alice\rq s bit strings tend to be more random than Bob\rq s ones.
This suggests asymmetry between Alice\rq s apparatus and Bob\rq s due to, e.g., systematic errors since the ideal Bell test has the exchange symmetry between Alice and Bob.
Note that one cannot extract information about such asymmetry by using $S$, which only quantifies the correlation between Alice and Bob. 
Note that the maximal value of the LZ complexity in our analysis is in accordance with our theoretical consideration in Sec.~\ref{sec:LZ_complexity}.
More precisely, all our analyses satisfy $K(N)\lesssim1.55$, which is accordance with $K(N)<\frac{1}{1-\epsilon_N}\approx2.652$ implied by Ineq.~(\ref{eq:bound_on_c(S)})  for the bit length $90000$, an upper bound length of the mixed bit string $c^N$ in our coincidence data.

%%%%%%%%%%
\begin{figure}[t]
%\begin{minipage}[a]{.8\linewidth}
 \includegraphics[width=.5\textwidth]{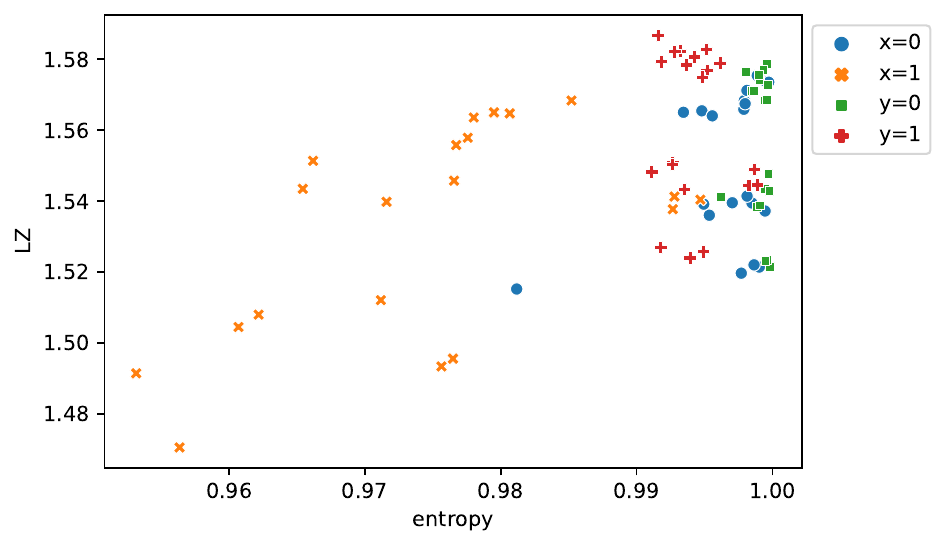}\\
%  \subcaption{LZ complexity.}
 % \label{LZvsEnt}
 % \end{minipage}
%\begin{minipage}[a]{.8\linewidth}
 \includegraphics[width=.5\textwidth]{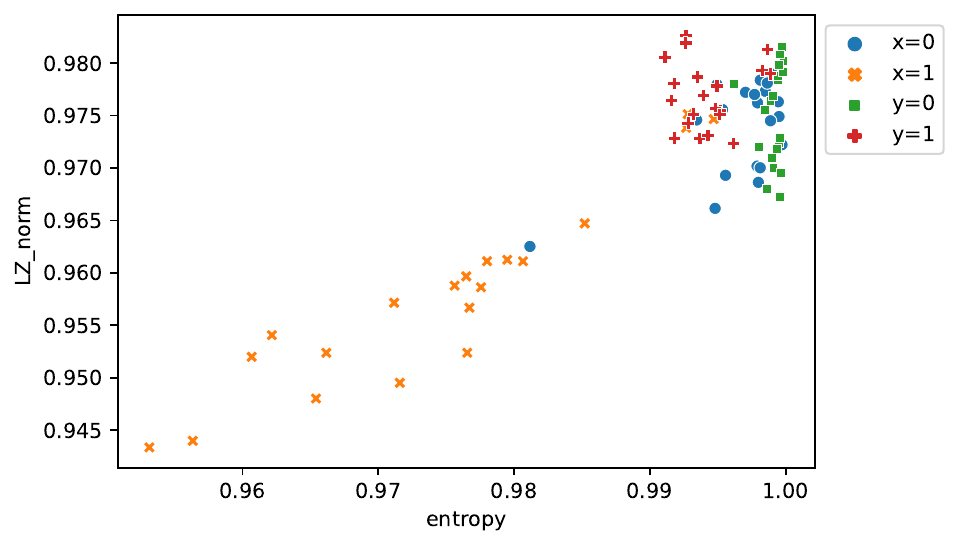}
 % \subcaption{Normalized LZ complexity.}
 % \label{LZnormalvsEnt}
 % \end{minipage}
%(a)  \includegraphics[width=0.4\textwidth]{q2_conditional_flat_postprocess1_clifford_deviation.pdf}
%\\
%(b)    \includegraphics[width=0.4\textwidth]{q2_conditional_flat_postprocess2_clifford_deviation.pdf}
  \caption{The scatter plots between the randomness measures and the entropy: (Above) the LZ complexity, (Below) the normalized LZ complexity. The outlier coming from longdist4 is omitted. The bit strings at the Bob\rq s station are more random than those at the Alice\rq s station. %The mixed bit strings are less random if evaluated by the LZ measure.
  }
  \label{Ent_scatter}
\end{figure}
%%%%%%%%%%

Above we analyzed Alice and Bob individually, but we also analyzed their mixed bit strings. 
We found that the relationship between the randomness of Alice's (or Bob's) individual bits and that of the mixed bit sequence appears to be different depending on the measure.
This may reflect the fact that the randomness measures we use, although normalized, is still weakly dependent on length.
%In addition to the above comparison between Alice\rq s and Bob\rq s  bit strings, Fig.~\ref{boxes} also shows mutually distinct results for their mixed bit strings:
The mixed bit string is statistically the least random if measured with the LZ complexity, whereas the most random with the normalized LZ complexity.
Moreover, the Borel normality suggests that its randomness is in-between.

We note that the mixed bit sequence should be the least random in the algorithmic sense, since $S>2$ for many of the longdist data suggests that every bit pair $(a_i,b_i)$ in the mixed bit strings has the strong correlation not described by any LHVTs, resulting in a decrease of the algorithmic randomness.
This suggests that perhaps the LZ complexity most captures the latent correlations inherent in the mixed bit strings obtained from the Bell tests.

%\textcolor{blue}{
The data obtained from ChaCha20 successfully imitate those from the Innsbruck experiment.
More precisely, the interquartile range (IQR, the box heights in the box-and-whisker plots) are almost the same as those of the associated Innsbruck data.
%}

%We first check whether the (normalized) LZ complexity correlates with the entropy of the outcomes occurring.
In the rest of this subsection, we reexamine the anti-correlation suggested in \cite{PhysRevA.98.042131}.
We first check the correlation between the entropy and the (normalized) LZ complexity.
Figure \ref{Ent_scatter} shows that given the conditional probability distributions $p(a|x)$, $p(b|y)$, their entropies 
\begin{eqnarray}
H(A|X=i)&=-\sum_{a}p(a|X=i)\log_2p(a|X=i),\nonumber\\
H(B|Y=i)&=-\sum_{b}p(b|Y=i)\log_2p(b|Y=i)
\end{eqnarray}
for $i=0,1$ correlate with the (normalized) LZ complexity.
This suggests that the LZ complexity can be understood in terms of entropy.
Moreover, we observe that the entropy for the case $x=0$ and $y=0,1$ ranges from 0.98 to 1, whereas that for $x=1$ does from 0.95 to 1.
Since the initial state in Weihs\rq~experiment is the singlet state, its reduced density matrix is completely mixed, leading to the equal weight probability distribution.
Therefore, the large deviation of the entropy when $x=1$ implies the existence of (systematic) errors in Alice\rq s station.
%$p_{\beta^N_{(k)}}(0)$ of finding $0$ in the bit strings $\alpha^N_{(k)}$ and $\beta^N_{(k)}$ deviate from the ideal value 1/2.

%%%%%%%%%%
\begin{figure}[]
%\begin{minipage}[a]{1\linewidth}
  \includegraphics[width=0.5\textwidth]{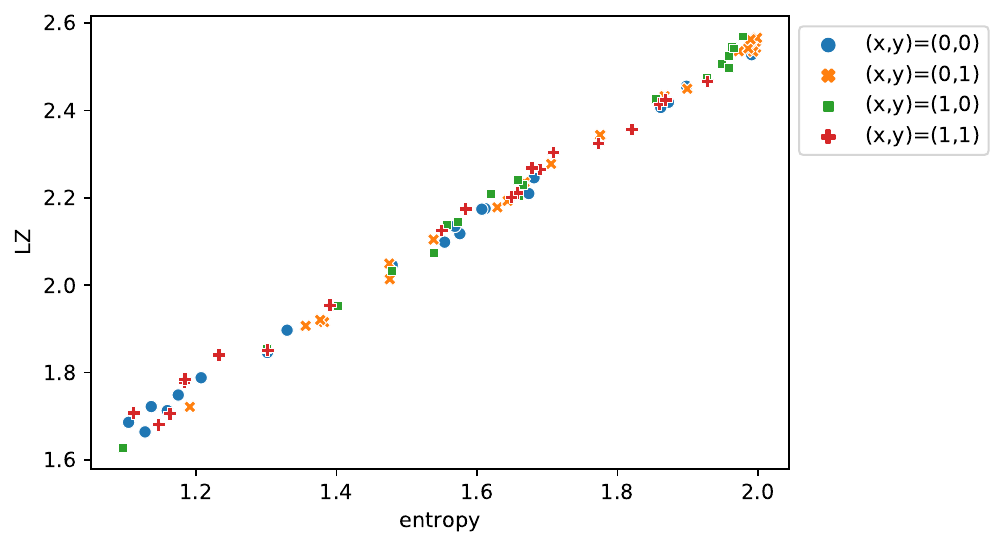}
%  \subcaption{LZ complexity.}
%  \label{jointLZvsEnt}
%  \end{minipage}
%(a)  \includegraphics[width=0.4\textwidth]{q2_conditional_flat_postprocess1_clifford_deviation.pdf}
%\\
%(b)    \includegraphics[width=0.4\textwidth]{q2_conditional_flat_postprocess2_clifford_deviation.pdf}
  \caption{The scatter plot between the randomness measures and the entropy of the joint probability distribution. The outlier coming from longdist4 is omitted.
  }
  \label{Ent_scatter_joint}
\end{figure}
%%%%%%%%%%

%%%%%%%%%%%%%%%%%%%%%%%%%%%%%%
\begin{table*}[t]
    \centering
    \caption{The correlation analysis of the bit strings with the stratification of $N$. The Pearson correlation coefficient $r$ between the non-locality and randomness measures is calculated.}
    \label{pearson}
    \begin{tabular}{ccc|ccc}
    \hline
        Stratification & Bit sequence & Randomness measure & $r$ & $p$-value & 95\% coincidence interval  \\ \hline\hline
         		&   & LZ & 0.450 & 0.192 & [-0.251, 0.841]\\ 
         	& Alice  & Normalized LZ & 0.405 & 0.246 & [-0.302, 0.824]\\ 
         		&   & Borel normality & -0.313 & 0.379 & [-0.787, 0.394]\\ 
         		&   & LZ & 0.092 & 0.801 & [-0.571, 0.682]\\ 
        $N<20000$ & Bob  & Normalized LZ & 0.115 & 0.751 & [-0.555, 0.694]\\ 
         	(10 samples)	&   & Borel normality & -0.080 & 0.825 & [-0.676, 0.579]\\ 
         		&   & LZ & 0.448 & 0.194 & [-0.253, 0.841]\\ 
         	& Mixed  & Normalized LZ & 0.325 & 0.360 & [-0.383, 0.792]\\ 
		         &   & Borel normality & -0.222 & 0.538 & [-0.747, 0.474]\\ \hline
		         &   & LZ & 0.395 & 0.381 & [-0.510, 0.885]\\
	         & Alice  & Normalized LZ & 0.437 & 0.327 & [-0.471, 0.895]\\ 
		         &   & Borel normality & -0.054 & 0.909 & [-0.775, 0.729]\\ 
		         &   & LZ & 0.047 & 0.921 & [-0.732, 0.773]\\ 
        $20000\le N < 40000$ & Bob  & Normalized LZ & -0.014 & 0.977 & [-0.759, 0.747]\\ 
		(7 samples)         &   & Borel normality & -0.178 & 0.703 & [-0.821, 0.664]\\ 
		         &   & LZ & -0.240 & 0.604 & [-0.841, 0.626]\\ 
	         & Mixed  & Normalized LZ & -0.515 & 0.237 & [-0.914, 0.389]\\ 
		         &   & Borel normality & -0.236 & 0.610 & [-0.840, 0.629]\\ \hline 
		         &   & LZ & -0.293 & 0.811 & [-1, 1]\\ 
	         & Alice  & Normalized LZ & -0.572 & 0.612 & [-1, 1]\\ 
		         &   & Borel normality & 0.075 & 0.952 & [-1, 1]\\ 
		         &   & LZ & 0.845 & 0.359 & [-1, 1]\\ 
        $N\ge 40000$  & Bob  & Normalized LZ & 0.879 & 0.316 & [-1, 1]\\ 
	(3 samples)	         &   & Borel normality & -0.975 & 0.144 & [-1, 1]\\ 
		         &   & LZ & 0.684 & 0.520 & [-1, 1]\\ 
	         & Mixed  & Normalized LZ & 0.665 & 0.537 & [-1, 1]\\ 
		         &   & Borel normality & -0.601 & 0.589 & [-1, 1]\\ 
         \hline 
    \end{tabular}
\end{table*}
%%%%%%%%%%%%%%%%%%%%%%%%%%%%%%

The above correlation can be observed in the entropy of the joint probability distribution
\begin{eqnarray}
&&H(AB|X=i, Y=j)\\
&&=-\sum_{a,b}p(ab|X=i, Y=j)\log_2p(ab|X=i, Y=j).\nonumber
\end{eqnarray}
Indeed, Fig.~\ref{Ent_scatter_joint} shows a correlation relation stronger than that shown in Fig.~\ref{Ent_scatter}.
Therefore we can conclude that the LZ complexity is explainable by the entropy of joint probability distribution for the Innsbruck experiment.

In contrast to entropy, we cannot say an anti-correlation exists between non-locality and randomness.
To argue this in detail, recall that the maximum of the LZ complexity depends on $N$ as shown in Sec.~\ref{measures}, and thereby the possible confounder is the number of the coincidence events $N$.
From this observation, we evaluated the Pearson correlation coefficients $r$ between $S$ and the randomness measures with the use of the stratification of $N$ into the following three classes: i) $N<20000$, ii) $20000\le N < 40000$, and iii) $N\ge 40000$.

Table \ref{pearson} shows the correlation coefficients $r$ with its $p$-values and 95\%  coincidence intervals.
We find that $r>0$ for the (normalized) LZ complexity in general whereas $r<0$ for the Borel normality except for the mixed bit sequence with $20000\le N<40000$.
This is consistent with the fact that the (normalized) LZ complexity becomes larger when the given sequence is more random, whereas the Borel normality is becoming smaller.
Besides, we observe that the correlation coefficient for Alice\rq s bit strings is comparable with that for the mixed bit strings,  except in the case $N\ge40000$.
Thus, we can conclude that the degree of the randomness (weakly) correlates with the non-locality for  Alice\rq s bit strings for $N<20000$ and $20000\le N<40000$.
This also holds for the mixed bit strings for $N\le20000$.
These observations are inconsistent with  \cite{PhysRevA.98.042131}.
On the other hand, all the $p$-values are more than 0.05, implying that it is hard to deduce the anti-correlation for the parent population if we see the longdist dataset as samples thereof.
Indeed, all the 95\% coincidence intervals of $r$ contain $r=0$.

It is of important to find the typical sample size to determine whether the non-locality (anti-)correlates with the degree of randomness in the parent population.
For example, let us take $r=0.45$, which is the correlation coefficient between $S$ and the LZ complexity of Alice\rq s bit strings for $N\le20000$.
According to the standard procedure given in \cite{cohen1988spa}, the sample size is 36 to test the hypothesis $r=0.45$ with significance level $\alpha=0.05$ (two-tailored) and power $1-\beta=0.8$.
Moreover, sample size estimation with a pre-specified confidence interval requires a larger sample size:
we need the sample size 64, 111, and 247 to test $r=0.45$ with desired 98\% confidence interval half-width 0.20, 0.15, and 0.10, respectively \cite{TQMP10-2-124}.
On the other hand, if we take $r=0.1$, which is the correlation coefficient between $S$ and the LZ complexity of Bob\rq s bit strings for $N\le20000$, then the required sample size is 782 for the test with significance level $\alpha=0.05$ (two-tailored) and power $1-\beta=0.8$.

%%%%%%%%%%%%%%%%%%%%%%%%%%%%%%
\begin{table*}[]
    \centering
        \caption{Numerics of the measures.}
    \begin{tabular}{crrrrrrrrrrr}
    \hline
         &  &  & \multicolumn{3}{c}{LZ} & \multicolumn{3}{c}{Normalized LZ} & \multicolumn{3}{c}{Borel normality}  \\
                  \cmidrule(lr){4-6}
                  \cmidrule(lr){7-9}
                  \cmidrule(lr){10-12}
        filename & $N$ & $S$ & Alice & Bob & Mixed & Alice & Bob & Mixed & Alice & Bob & Mixed  \\ \hline\hline
        longdist0 & 17374  & 2.518 & 1.531 & 1.535 & 1.503 & 0.973 & 0.976 & 0.981 & 0.595 & 0.182 & 0.424 \\ 
        longdist1 & 17640  & 2.629 & 1.531 & 1.534 & 1.499 & 0.974 & 0.975 & 0.978 & 0.462 & 0.166 & 0.331 \\ 
        longdist2 & 27708  & 1.978 & 1.507 & 1.511 & 1.482 & 0.979 & 0.982 & 0.983 & 0.483 & 0.207 & 0.398 \\ 
        longdist3 & 27620  & 2.660 & 1.510 & 1.511 & 1.480 & 0.981 & 0.981 & 0.982 & 0.499 & 0.221 & 0.364 \\ 
        longdist4 & 1112  & 1.810 & 0.819 & 1.665 & 1.380 & 0.464 & 0.943 & 0.812 & 7.158 & 1.110 & 3.690 \\ 
        longdist5 & 27420  & 2.622 & 1.508 & 1.510 & 1.479 & 0.980 & 0.981 & 0.981 & 0.517 & 0.203 & 0.386 \\ \hline
        longdist10 & 28248  & 2.283 & 1.503 & 1.507 & 1.477 & 0.977 & 0.980 & 0.981 & 0.641 & 0.370 & 0.516 \\ 
        longdist11 & 28818  & 2.405 & 1.507 & 1.508 & 1.477 & 0.980 & 0.981 & 0.981 & 0.649 & 0.278 & 0.515 \\ 
        longdist12 & 28718  & 2.393 & 1.503 & 1.508 & 1.477 & 0.978 & 0.981 & 0.982 & 0.636 & 0.254 & 0.521 \\ 
        longdist13 & 28513  & 2.357 & 1.502 & 1.506 & 1.477 & 0.976 & 0.979 & 0.981 & 0.784 & 0.335 & 0.559 \\ \hline
        longdist20 & 43315  & 2.058 & 1.491 & 1.489 & 1.463 & 0.977 & 0.976 & 0.978 & 0.166 & 0.341 & 0.272 \\ 
        longdist22 & 42797  & 2.186 & 1.494 & 1.490 & 1.464 & 0.979 & 0.976 & 0.979 & 0.071 & 0.342 & 0.223 \\ 
        longdist23 & 42778  & 2.632 & 1.491 & 1.491 & 1.464 & 0.976 & 0.976 & 0.979 & 0.143 & 0.265 & 0.228 \\ \hline
        longdist30 & 15306  & 2.104 & 1.524 & 1.539 & 1.500 & 0.968 & 0.977 & 0.977 & 1.095 & 0.311 & 0.745 \\ 
        longdist31 & 15122  & 2.637 & 1.517 & 1.542 & 1.500 & 0.963 & 0.979 & 0.977 & 1.168 & 0.283 & 0.756 \\ 
        longdist32 & 14294  & 2.703 & 1.521 & 1.538 & 1.499 & 0.964 & 0.976 & 0.975 & 1.184 & 0.562 & 0.874 \\ 
        longdist33 & 15113  & 2.055 & 1.519 & 1.541 & 1.498 & 0.964 & 0.978 & 0.975 & 1.170 & 0.404 & 0.794 \\ 
        longdist34 & 14824  & 1.949 & 1.520 & 1.536 & 1.499 & 0.964 & 0.974 & 0.976 & 1.133 & 0.361 & 0.764 \\ 
        longdist35 & 14573  & 2.728 & 1.523 & 1.542 & 1.500 & 0.967 & 0.979 & 0.976 & 1.198 & 0.338 & 0.745 \\ 
        longdist36 & 14571  & 2.720 & 1.527 & 1.541 & 1.502 & 0.969 & 0.978 & 0.977 & 1.020 & 0.378 & 0.743 \\ 
        longdist37 & 14673  & 1.953 & 1.518 & 1.540 & 1.498 & 0.963 & 0.978 & 0.975 & 1.241 & 0.366 & 0.743 \\ \hline 
    \end{tabular}
    \label{numerics}
\end{table*}
%%%%%%%%%%%%%%%%%%%%%%%%%%%%%%

%%%%%%%%%%%%%%%%%%%%%%%%%%%%%%
\section{Conclusion}
\label{conc}
%%%%%%%%%%%%%%%%%%%%%%%%%%%%%%

In this paper, we presented a no-go theorems (and its generalization) showing the contradiction between the existence of the CPRNGs and that of the efficiently computable randomness measures distinguishing the outcomes from the PRNGs and those from ECS systems. 
Given the success of modern cryptography which hinges on the presumed existence of the CPRNGs, it seems difficult to deny the assumption in practice, despite the fact that by definition PRNGs (including CPRNGs) are less random than ECS systems (see Section \ref{sec:PRNG}).   Accordingly, we could say that there are no efficiently computable randomness measures capable of distinguishing the outcomes of PRNGs from those of the ECS systems.  It follows that if the QRNG under consideration is ECS then these exists no random measures to distinguish the QRNG from a properly chosen PRNG.    

In addition, we gave two comparative analyses between the PRNGs and the actual data on the quantum coin tosses and the Innsbruck experiment.
In both analyses, we found that the average values of the algorithmic complexity measures for the PRNGs are indistinguishable from those for the QRNGs.  This is an assuring result and is consistent with our theorem.
In the meantime, we reconsidered the former analyses \cite{PhysRevA.98.042131} which erroneously reported an anti-correlation between the randomness and non-locality in the Innsbruck experiment.  In contrast, we found that the randomness measures show no such correlation between locality and randomness in general.  Instead, a clear correlation is found between the randomness measures and (joint) binary entropy, suggesting the existence of possible systematic errors in Alice\rq s apparatus.  

The fact that our theorem can be applied only to ECS quantum systems indicates that for QRNG the property of ECS may be an important characteristic when it comes to the study of the nature of randomness.  Since this connection has been largely unnoticed in the existing literature \cite{PhysRevA.82.022102,Solis_2015,e20110886,PhysRevA.98.042131,Abbott_2019}, it would be of interest to find quantum systems which are not ECS and examine if their outcomes exhibit any signature different from PRNG.

%\acknowledgments
\section*{Acknowledgment}
We thank G. Weihs for providing the dataset of the Innsbruck experiment.
This work was supported by MEXT Quantum Leap Flagship Program (MEXT Q-LEAP) Grant Number JPMXS0120319794, and by JSPS KAKENHI Grant Number 20H01906 and JP22K13970.

%%%%%%%%%%%%%%%%%%%%%%%%%%%%%%
\appendix
%%%%%%%%%%%%%%%%%%%%%%%%%%%%%%

%%%%%%%%%%%%%%%%%%%%%%%%%%%%%%
\section{Numerical features of the measures in the longdist* samples}
\label{measure_data}
%%%%%%%%%%%%%%%%%%%%%%%%%%%%%%

%\textcolor{blue}{
Table~\ref{numerics} shows the statistics of the longdist datasets in the Innsbruck experiment.
The number of the pairs of the coincident outcomes $N$ obtained in our analyses is almost identical to that obtained in the original analyses by Weihs \cite{weihs_gregor_1998_7185335} with a difference of about ${\cal O}(10)$, except the sample longdist4. 
In contrast, $N$ given in \cite{PhysRevA.98.042131} differs from the original analyses by Weihs about ${\cal O}(10^3)$.
%}

%\textcolor{blue}{
The exceptional behavior of the sample longdist4 is presumably due to a difference between the stored and actual threshold values.
Indeed, in the original analyses \cite{weihs_gregor_1998_7185335}, the stored threshold value and $N$ for longdist4 are completely identical to those for longdist3, which is unlikely.
Taking this possibility into account, we regarded longdist4 as an outlier and excluded it from our analyses.
%}

\section{Proof of Theorem \ref{thm:main_gen}}
\label{sec:proof_thm}

In this section, we prove Theorem \ref{thm:main_gen}.

We begin by giving an overview of the proof.
The basic idea here is the same as in the case of the biased bit generator in Section \ref{sec:simple_case} (see Fig. \ref{fig:proofsketch}).

As we assumed that system $A$ is ECS, there exists a probabilistic algorithm $A'$ which simulates $A$ with arbitrary accuracy.
Let $\bar{A}(1^n)$ be the algorithm obtained by repeating $A'(1^n)$ $n$ times.  
Since $\bar{A}(1^n)$ is also a probabilistic algorithm, it can be described as a deterministic algorithm having a TRNG $U^{l_G(n)}$ as an auxiliary input which $\bar{A}(1^n)$ uses for its internal coin tossing, with $l_G(n)$ being a polynomial of $n$.
By replacing the auxiliary input $U^{l_G(n)}$ by a CPRNG $G(U^k)$ ($\in\{0,1\}^{l_G(k)}$), we obtain the desired algorithm $B(U^k)$.

To see that $B(U^k)$ thus obtained is indeed indistinguishable from the original physical system $A^{l(k)}$, first note that $\bar{A}(1^{l(k)})$ simulates the output of $A^{l(k)}$ with arbitrary accuracy, due to the ECS property of $\bar{A}(1^n)$.
Also note that, by definition of the CPRNG, TRNG $U^{l_G(k)}$ and CPRNG $G(U^k)$ are indistinguishable from each other, and thus their outputs followed by an application of an identical algorithm $\bar{A}(1^{l(k)})$ are also indistinguishable; namely, the outputs of $\bar{A}(1^{l(k)})$ and $B(U^k)$ are indistinguishable.
As a result, the outputs of $A^{l(k)}$ and $B(U^k)$ are indistinguishable.

We will below elaborate this discussion.

\begin{figure}[t]
  \includegraphics[width=\linewidth]{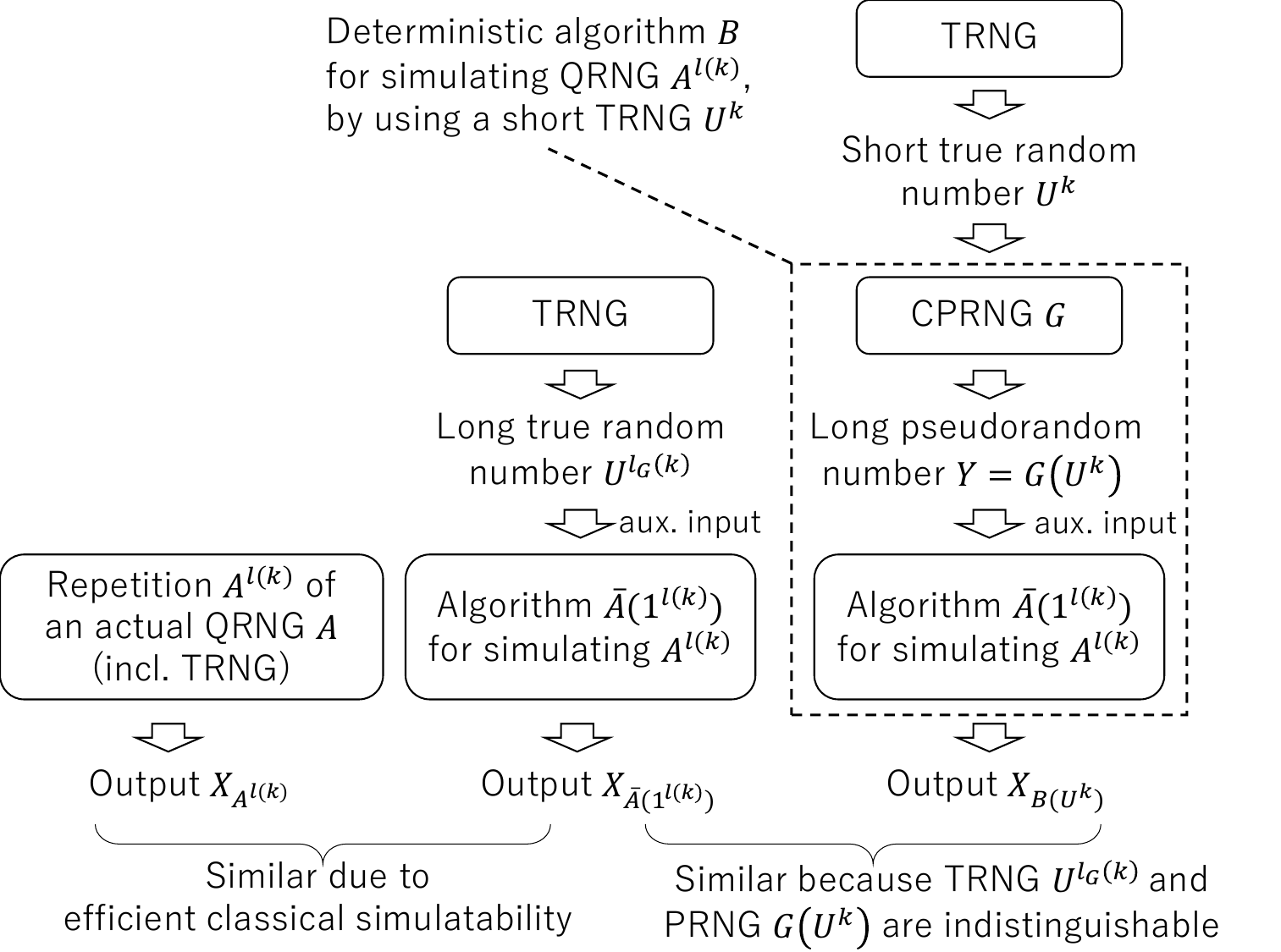}
  \caption{An overview of the proof of Theorem \ref{thm:main_gen}.
  Algorithm $\bar{A}(1^n)$ is the one in which the classical simulation algorithm $A'(1^n)$ of $A$ is repeated  $n$ times.
  In the above figure, we set $n=l(k)$.
  }
  \label{fig:proofsketch}
\end{figure}

\subsection{Construction of algorithm $B$}

Following Definition \ref{def:efficient_simulation}, fix an efficient classical simulation of $A$, and denote it by $A'$.

Let $\bar{A}(1^n)$ be the algorithm obtained by repeating $A'(1^n)$ $n$ times.
Since $\bar{A}$ is a polynomial-time algorithm, the running time of $\bar{A}(1^n)$ can be bounded from above by a polynomial $g(n)$.
Without loss of generality, we may assume that $g(n)$ is monotonically increasing.

Recall that probabilistic algorithms in general can be described as a deterministic algorithm having a random tape as auxiliary inputs (see, e.g., Ref. \cite{goldreich_2001}, Section 1.3.2.2).
In this description, one can regard $\bar{A}(1^n)$ as a deterministic algorithm having TRNG $U^{l_G(n)}$ as an auxiliary input (to be used for its internal coin tossing), and halts within time $g(n)$.

We define algorithm $B$ by replacing $\bar{A}(1^n)$'s auxiliary input TRNG $U^{l_G(n)}$ by a CPRNG $G(U^k)$.
More precisely, choose a CPRNG $G:\{0,1\}^k\to\{0,1\}^{l_G(k)}$ with an expansion factor $l_G(k)=g(l(k))$ (the existence of such $G$ is guranteed, e.g.,  by Ref. \cite{goldreich_2001}, Sections 3.3.1 and 3.3.2, or Ref. \cite{books/crc/KatzLindell2020}, Theorem 8.7).
Then define a deteriministic algorithm $B:\{0,1\}^*\to{\cal X}^*$ as
\begin{enumerate}
\item On input a random seed $s\in\{0,1\}^k$, $B$ inputs $s$ to $G$  and obtains the output $r\in\{0,1\}^{l_G(k)}$.
\item $B$ starts algorithm $\bar{A}$ with setting $1^{l(k)}$ as the input and $s$ as the auxiliary input, and obtains the output $x=(x_1,\dots,x_{l(k)})\in{\cal X}^{l(k)}$.
$B$ outputs $x$.
\end{enumerate}

\subsection{Indistinguishability of $A^{l(k)}$ and $\bar{A}(1^{l(k)})$}

In the above setting, the variational distance between the outputs of $A^n$ and $\bar{A}(1^n)$ can be bounded as
\begin{align}
&\|P(X_{A^n})-P(X_{\bar{A}(1^n)})\|_1\nonumber\\
&\le\|P(X_{A^n}^n)-P(X_AX_{\bar{A}(1^{n-1})})\|_1\nonumber\\
&\quad+\|P(X_AX_{\bar{A}(1^{n-1})})-P(X_{\bar{A}(1^n)}^n)\|_1\nonumber\\
&=\|P(X_A)\|\|P(X_{A_{n-1}})-P(X_{\bar{A}(1^{n-1})})\|_1\nonumber\\
&\quad+\|P(X_A)-P(X_A')\|_1\|P(X_{\bar{A}(1^{n-1})})\|_1\nonumber\\
&\le\|P(X_{A_{n-1}})-P(X_{\bar{A}(1^{n-1})})\|_1+2^{-n}\nonumber\\
&\le\cdots\le n2^{-n}\le {\rm negl}(n),
\label{eq:A'n_simlatability}
\end{align}
where $\|P(X)\|_1$ denotes the trace norm of the probability distribution of $X$, $\|P(X)\|_1:=\frac12\sum_{x\in{\cal X}}P(X=x)$.
Thus by letting $n=l(k)$ (according to the setting of the previous subsection), we have
\begin{align}
\left\|P(X_{A^{l(k)}})-P(X_{\bar{A}(1^{l(k)})})\right\|_1\le {\rm negl}(l(k))\le {\rm negl}(k).
\end{align}
Then by using the monotonicity of the variational distance (see, e.g., \cite{nielsen_chuang_2010}, Theorem 9.2), we have
\begin{align}
&\left|P(D_I(I(X_{A^{l(k)}}))=1)
-P(D_I(I(X_{\bar{A}(1^{l(k)})}))=1)\right|\nonumber\\
&\le  {\rm negl}(k).
\label{eq:proof_negl1}
\end{align}

\subsection{Indistinguishability of $\bar{A}(1^{l(k)})$ and $B(U^k)$}

Next we define a distinguisher algorithm $D:\{0,1\}^*\to\{0,1\}$ for CPRNG having an expansion factor $l_G(k)$ as
\begin{enumerate}
\item On input $r\in\{0,1\}^*$, $D$ finds $k\in\N$ satisfying $l_G(k)=|r|$.
If such $k$ is not found, $D$ halts.
\item $D$ starts $\bar{A}$ with setting $1^{l(k)}$ as input and $r$ as the auxiliary input, and obtains the output $x^{l(k)}$.
\item $D$ inputs $x^{l(k)}$ to $D_I(I(\dots))$ and obtains the output $b\in\{0,1\}$.
$D$ outputs $b$.
\end{enumerate}
To see that this algorithm is indeed efficient, note that each step above can be executed within a polynomial time of $k$, and that $k$ is always less than the input length $l_G(k)$ of $D$, since $l_G(k)$ is an expansion factor of $G$.

As we assume that $G$ is a CPRNG, we have
\begin{align}
\left|P(D(G(U^k))=1)-P(D(U^{l_G(k)})=1)\right|\le {\rm negl}(k).
\label{eq:thm_indistinguishability}
\end{align}

Also, by definition, the situations where one inputs $U^{l_G(k)}$ to $D$, and where one input $1^{l(k)}$ to $\bar{A}$ are identical.
Thus we have
\begin{align}
P(D(U^{l_G(k)})=1)=P(D(I(X_{\bar{A}(1^{l(k)})}))=1).
\label{eq:thm_D_bar1}
\end{align}
Similarly, since the situations where one inputs $G(U^k)$ to $D$, and where one inputs $U^k$ to $B$ are identical, 
\begin{align}
P(D(G(U^k))=1]=P[D(I((B(U^k)))=1).
\label{eq:thm_D_bar2}
\end{align}

Then by combining relations (\ref{eq:thm_indistinguishability}), (\ref{eq:thm_D_bar1}), and (\ref{eq:thm_D_bar2}), we have
\begin{align}
&\left|P(D(I(X_{\bar{A}(1^{l(k)})}))=1)-P(D(I((B(U_k)))=1)\right|\nonumber\\
&\le {\rm negl}(k).
\label{eq:proof_negl2}
\end{align}

\subsection{Triangle inequality}

Finally, by applying the triangle inequality to inequalities (\ref{eq:proof_negl1}) and (\ref{eq:proof_negl2}), we obtain (\ref{eq:thm_equation}).

\bibliography{ref_qrand}

%--------------------------%
%\begin{thebibliography}{99}
%--------------------------%
%
%\bibitem{Jaynes03}
%E. T. Jaynes,
%{\it Probability Theory: The Logic of Science}
%(Cambridge University Press, Cambridge, 2003).
%--------------------%
%\end{thebibliography}

\end{document}